\documentclass{scrartcl}
\typearea[10mm]{12}
\usepackage[utf8x]{inputenc}
\usepackage[OT4]{fontenc}
\usepackage{tensor}
\usepackage{amsmath,amsthm,latexsym,amssymb,amsfonts}
\usepackage{bbm}
\usepackage{subcaption}
\usepackage{hyperref}
\usepackage{mathtools}
\newcommand{\triad}{e}
\newcommand{\connection}{\omega}

\newcommand{\spacetime}{M}
\newcommand{\bundle}{P}
\newcommand{\abundle}{{\rm Ad}(\bundle)}
\newcommand{\Tr}[1]{{\rm Tr}\left(#1\right)}
\newcommand{\curvature}{F}
\newcommand{\intmetric}{\eta}
\newcommand{\metric}{g}

\newcommand{\complex}{\mathfrak{K}}

\newcommand{\simplex}{\Delta}

\newcommand{\manifold}{M}
\newcommand{\gluing}[1][\complex]{|#1|}
\newcommand{\isfaceof}{\prec}
\newcommand{\gsimplex}[1][i]{|\simplex_{#1}|}
\newcommand{\gsimplexm}[2]{|#1|_{\mu_{#2}}}
\newcommand{\gsimplexmp}[2]{|#1|_{\mu'_{#2}}}

\newcommand{\rotate}[1]{{\rm R}_{#1}}
\newcommand{\hyperplane}{\Sigma}
\newcommand{\dbundle}{{\rm SO(}\intmetric{\rm )}\times_{\rho} \mathbb{R}^n}
\newcommand{\tmap}[1][ij]{g_{#1}}
\newcommand{\vol}[1][n]{{\rm vol}_{#1}}
\newcommand{\curvaturematrix}{\mathcal{F}}
\newcommand{\edge}{l}
\newcommand{\length}{\ell}
\newcommand{\lattice}{\Lambda}
\newcommand{\latticeconst}{\epsilon}
\newcommand{\B}{B}
\newcommand{\translation}{T}
\newcommand{\area}{A}
\newcommand{\deficit}{\varepsilon}
\newcommand{\dihedral}{\theta}
\newcommand{\identity}{\mathbbm{1}}
\newcommand{\sgn}{{\rm sgn}}
\newcommand{\bigo}{\mathcal{O}}
\newcommand{\ine}{\iota}
\newcommand{\oute}{o}
\newcommand{\SL}{{\rm SL}(2,\mathbb{C})}
\DeclareMathOperator{\D}{D}
\DeclareMathOperator{\iu}{i}
\DeclareMathOperator{\id}{id}
\DeclareMathOperator{\interior}{int}
\DeclareMathOperator{\conv}{conv}
\DeclareMathOperator{\card}{\#}
\DeclareMathOperator{\diam}{diam}
\DeclareMathOperator{\mesh}{mesh}
\DeclareMathOperator{\diag}{diag}
\DeclareOldFontCommand{\rm}{\normalfont\rmfamily}{\mathrm}
\theoremstyle{theorem}
\newtheorem{thm}{Theorem}
\newtheorem{prop}{Proposition}

\theoremstyle{definition}
\newtheorem{ex}{Example}
\newtheorem{rmk}{Remark}
\begin{document}
\title{Relation between Regge calculus and BF theory on manifolds with defects}
\author{Marcin Kisielowski\footnote{Marcin.Kisielowski@gmail.com}\\
\textit{\footnotesize ${}^1$ Institute for Quantum Gravity, Chair for Theoretical Physics III,}\\ \textit{\footnotesize University of Erlangen-Nürnberg, Staudtstraße 7 / B2,
91058 Erlangen, Germany}\\
\textit{\footnotesize  ${}^2$ Instytut Fizyki Teoretycznej, Uniwersytet Warszawski,}\\
\textit{\footnotesize ul. Pasteura 5, 02-093 Warsaw, Poland}}
\date{ }
\maketitle
\begin{abstract}
In Regge calculus the space-time manifold is approximated by certain abstract simplicial complex, called a pseudo-manifold, and the metric is approximated by an assignment of a length to each 1-simplex. In this paper for each pseudomanifold we construct a smooth manifold which we call a manifold with defects. This manifold emerges from the purely combinatorial simplicial complex as a result of gluing geometric realizations of its n-simplices followed by removing the simplices of dimension n-2. The Regge geometry is encoded in a boundary data of a BF-theory on this manifold. We construct an action functional which coincides with the standard BF action for suitably regular manifolds with defects and fields. On the other hand, the action evaluated at solutions of the field equations satisfying certain boundary conditions coincides with an evaluation of the Regge action at Regge geometries defined by the boundary data. As a result we trade the degrees of freedom of Regge calculus for discrete degrees of freedom of topological BF theory. %This formulation leads to new interpretation of the Regge calculus. In the Plebański formulation the theory of General Relitivity is formulated as a BF theory with constraints that are imposed at each point of the manifold. The Regge calculus is obtained from the BF theory by imposing the constraints at the defects only.
\end{abstract}
%\begin{keyword}
%Regge calculus, BF theory, flat connections, manifold with defects, topological field theory
%\end{keyword}
%\maketitle
\section{Introduction}
\subsection{Motivation}
Our motivation for this research comes from the spin-foam approach to Quantum Gravity \cite{rovelli2014covariant,RovelliBook,Baezintro,PerezOldReview,PerezNewReview,Zakopane,EngleReview,LQG25,AshtekarRovelliReuterRev,ThiemannBook}. 
The spin-foam models can be thought to be quantum versions of Regge calculus. They are based on the observation made by Ponzano and Regge \cite{PRmodel} that an algebraic object called $6j$-symbol, appearing in recoupling theory of 4 quantum angular momenta, relates to the Regge action for 3-dimensional Euclidean gravity. This result has been extended to the physical case: 4-dimensional Lorentzian gravity. Such generalized symbol is called a vertex amplitude and is a basic building block of the so-called EPRL/FK spin-foam model \cite{EPRL,FK,BC,BCLorentz,BarrettAsymptotics,BRRmodel,SFLQG,cEPRL,OSM,EngleVertex,EngleVertexRev,BianchiHellmann}. The derivation of the spin-foam model is based on the Plebański formulation of Gravity as a theory of BF type \cite{Plebanski}. The spin-foam model is obtained by triangulating the space-time manifold, constructing a discrete path integral for BF theory and imposing the constraint on the 2-simplices of this triangulation. 

If we remove the 2-simplices of the triangulation, we obtain a manifold with defects. As claimed in \cite{BianchiLQGflat} spin-foam models can be interpreted as topological theories of BF type on such manifolds with defects. In fact, the vertex amplitude for the Lorentzian model with non-zero cosmological constant has been recently introduced as the $\SL$ Chern-Simons expectation value of certain Wilson-graph operator \cite{ChernSimonsSFI,ChernSimonsSFII,ChernSimonsSFIII,ChernSimonsSFIV}.

In this paper we study a classical relation between the (topological) BF theory \cite{BFI,BFII,BFIII,BFIV,BFV,BFVI} and Regge calculus \cite{Regge}. Let us consider a triangulation of space-time. The Einstein-Hilbert action evaluated at locally flat metrics singular at the 2-simplices of the triangulation coincides with the Regge action \cite{SorkinThesis} (see also \cite{FriedbergLee}). In fact we can remove the 2-simplicies (together with appropriate neighborhood), creating thus a manifold with defects, and evaluate the Einstein-Hilbert action with the Gibbons-Hawking-York boundary term \cite{GHYI,GHYII} at appropriate locally flat metric. When evaluated on such metrics the interior part of the action vanishes and the boundary term becomes the Regge action \cite{BianchiLQGflat}. This provides some classical justification for the interpretation of spin-foam models as topological theories of BF type on manifolds with defects. It is however not fully satisfactory. The first reason is that Einstein's theory of General Relativity is not a topological theory. When viewed as a theory of BF-type it involves constraints that are imposed not only at the boundary but also in the interior. If we relax these conditions and impose the constraints only at the boundary, we obtain a topological field theory (BF theory). In fact, we will show that the BF action evaluated at solutions of its field equations satisfying certain boundary conditions coincides with an evaluation of the Regge action at Regge geometries defined by the boundary data. Therefore we trade the topological degrees of freedom of the BF theory for the discrete degrees of freedom of Regge calculus. Second reason is that the locally flat metric is not the only solution of the Einstein's equations but this is the case for the BF theory. The third reason is that in the previous approach a starting point was a smooth space-time which is triangulated afterwards, whereas spin foams (and related to them Group Field Theories) refer to a combinatorial structure of abstract complexes and the space-time emerges as build from some basic building blocks \cite{GFTI,GFTII,GFTIII,GFTIV,GFTV,FeynmanSF}. Such combinatorial characterization gives a better control over possible histories of quantum gravitational field. In this paper we will start with an abstract simplicial complex equipped with certain geometric structure and from this data we will construct a smooth manifold with defects. This approach will be more general than the previous approach because each manifold with defects obtained by removing simplicial complexes from the triangulation can be obtained by our procedure but we will obtain also manifolds with defects that cannot be completed to manifolds without defects.
\subsection{Motivating example and our hypothesis}
The basic building block of the Regge geometries is the $\epsilon$-cone metric \cite{Regge}. The simplest example of such metric is the metric on a 2-dimensional cone with the apex removed, i.e. on $\mathbb{R}^2_*:=\mathbb{R}^2\slash \{0\}$, induced by the flat metric on $\mathbb{R}^3$:
$$
ds^2=\frac{1}{\sin^2\theta} dr^2+r^2d\varphi^2,
$$
where $0<2\theta\leq\pi$ is the apex angle of the cone. A generalization of this metric to three dimensions is the $\epsilon$-cone metric on $\manifold=\mathbb{R}^2_*\times [0,\ell]$:
$$
ds^2=dz^2+\frac{1}{\sin^2\theta} dr^2+r^2d\varphi^2.
$$
Let us switch to the triad formalism:
$$
\triad= \frac{1}{\sin \theta} dr\ \tau_1  + r d\varphi\ \tau_2 + dz\ \tau_3,
$$
where $\tau_i=-\iu \sigma_i, i\in\{1,2,3\}$ is a basis of the su(2) Lie algebra defined by the Pauli matrices $\sigma_1,\sigma_2,\sigma_3$. We denote by $\eta$ an invariant (positive-definite) scalar product on su(2) in which the basis is orthonormal. The $\epsilon$-cone metric is related to the internal metric $\eta$ by the standard relation:
$$
g(X,Y)=\eta(\triad(X),\triad(Y))
$$
for any $X,Y\in {\rm T}\manifold$. The connection 1-form is:
$$
\connection=\frac{1}{2} \sin\theta d\varphi \ \tau_3.
$$
Although the connection is flat
$$
\curvature=d\connection +[\connection,\connection]=0
$$
on $\manifold$ a holonomy around any loop $\gamma$ encircling the defect $\{(0,0,z):0 \leq z \leq \ell\}$ is non-trivial
$$
U_\gamma=\begin{bmatrix} \exp(-\iu \pi \sin \theta)&0\\ 0& \exp(\iu \pi \sin\theta)\end{bmatrix}.
$$
It is the rotation by an angle $2\pi-\epsilon$ preserving the direction of the defect. The angle $\epsilon=\pi \sin \theta$ is called a deficit angle. 

The fact that a holonomy of a flat connection on a non-simply connected manifold can be non-trivial is a basis of mathematical formulation of the Aharonov-Bohm effect \cite{AharonovBohm}. The defect models long and infinitely thin solenoid. Outside the solenoid the magnetic field (i.e. curvature form) vanishes but the magnetic flux through a surface cutting the solenoid transversely into two pieces is non-trivial. This can be interpreted as the fact that the magnetic field (i.e. curvature form) is concentrated on the defect:
\begin{equation}\label{eq:distributionalcurvature}
\curvature=\epsilon\, \delta(x,y)\, \tau_3 \, dx\wedge dy.
\end{equation}
Indeed, with this interpretation of the curvature, we recover the (abelian) Stokes' theorem:
$$
\exp(\int_{S_\gamma} \curvature)=U_\gamma,
$$
where $S_\gamma$ is any surface which boundary is $\gamma$.  

In three dimensions the triad coincides with the $B$-field and Palatini action coincides with the BF action:
$$
S[\triad,\connection]=\int_{\mathbb{R}^2\times [0,\ell]} \eta_{ij}\triad^i\wedge \curvature^j.
$$
Let us note that outside of the defect, i.e. on $\mathbb{R}^2_*\times [0,\ell]$, the configuration $(\triad,\connection)$ presented above satisfies the field equations:
$$
d\triad+\connection\wedge \triad=0,\quad \curvature=0.
$$
Although $\triad$ is only defined on the manifold with defect ($\mathbb{R}^2_*\times[0,\ell]$), we can use the extended definition of the curvature form \eqref{eq:distributionalcurvature} to extend the measure
$$
\eta_{ij} \triad^i\wedge \curvature^j
$$
to the manifold without defect ($\mathbb{R}^2\times [0,\ell]$). Remarkably, the evaluation at such configuration 
\begin{equation}
S[\triad,\connection]=\int_{\mathbb{R}^2\times [0,\ell]}  \epsilon\, \delta(x,y)\, dx\wedge dy\wedge dz= \ell\, \epsilon
\end{equation}
coincides with (a part of) the Regge action if we interpret the number $\ell$ as a length of an edge (the defect) and the angle $\epsilon$ as a deficit angle associated to this edge. We expect that it is a general feature: \emph{ BF action evaluated at solutions of its field equations satisfying certain boundary conditions coincides with an evaluation of the Regge action at Regge geometries defined by the boundary data.}
\subsection{Illustration of the construction of the action functional}
The example above is very simple and is missing some essential features of the construction presented in this paper. One of the simplifications is that the flat connection is reducible: the holonomy matrix is an element of the abelian $U(1)$ subgroup of the structure group $SU(2)$. Thanks to this abelian Stokes' theorem can be used. Another simplification is that the manifold with defect can be completed to a manifold without defect. We will use a more general definition of manifold with defects where this will not always be possible. A way out to deal with both problems is to define the action functional as a Riemann-like series. This solves the first problem because the surface ordered integral from the non-abelian version of the Stokes' theorem \cite{Arefeva} can be approximated by a regular surface integral for small loops. This infinitesimal version of the Stokes theorem has in fact been used in \cite{Arefeva} to derive the non-abelian Stokes theorem. We will solve the second problem by appropriately choosing the lattice: such that only the data on the manifold with defects will be used. Let us illustrate this idea on the following example. Let us consider a connection and a triad $\triad$ on $\manifold=\mathbb{R}^2_*\times[0,\ell]$. We do not assume that the connection is flat or trosionless. We introduce the following cubical lattice:
$$
\lattice=\latticeconst \left(\mathbb{Z}^3+(\frac{1}{2},\frac{1}{2},0)\right)\cap \manifold,
$$
where $\latticeconst$ is the lattice constant. We assume that $\latticeconst=\frac{\ell}{m}$ for some positive natural number $m$. For each point $p=\latticeconst( a+\frac{1}{2},b+\frac{1}{2},c)\in \lattice$ we consider a holonomy $U_{p;\mu\nu}$ along a plaquette (see for example \cite{QFTLattice}) containing the points:
$$
p, \quad p+\latticeconst\hat{\mu},\quad p+\latticeconst \hat{\mu}+\latticeconst \hat{\nu},\quad p+ \latticeconst \hat{\nu},
$$ 
where $\mu,\nu\in\{1,2,3\}$, $\hat{1}=(1,0,0),\,\hat{2}=(0,1,0),\,\hat{3}=(0,0,1)$.
Let $\curvaturematrix_{p;\mu\nu}$ be an element of the so(3) Lie algebra such that
\begin{equation}\label{eq:Fdef}
U_{p;\mu\nu}=\exp(\latticeconst^2\,\curvaturematrix_{p;\mu\nu}).
\end{equation}
\begin{figure}
\begin{subfigure}{0.45\textwidth}
 \centering
\includegraphics[width=.9\linewidth]{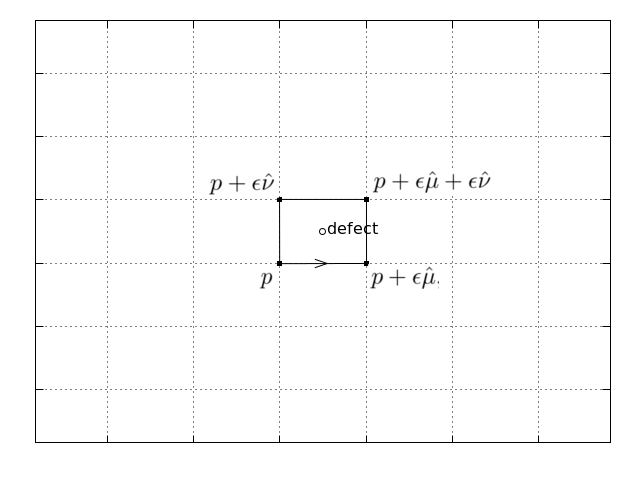}
  \caption{Section $z=0$ of the lattice. On this figure $p=\latticeconst(-\frac{1}{2},-\frac{1}{2},0)$, $\hat{\mu}=\hat{1}$, $\hat{\nu}=\hat{2}$. }
  \label{fig:lattice}
\end{subfigure}
\begin{subfigure}{0.45\textwidth}
 \centering
\includegraphics[width=.9\linewidth]{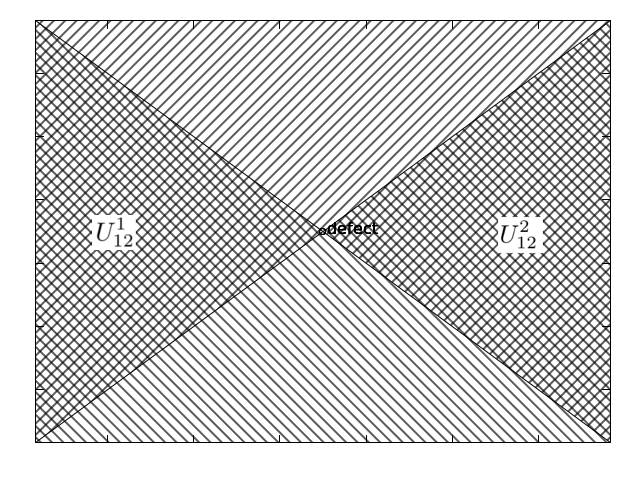}
  \caption{Covering of $\manifold=\mathbb{R}^2_*\times[0,\ell]$ with two open sets $U_1$ and $U_2$. The intersection $U_1\cap U_2$ is formed by two disjoint components $U^1_{12}$ and $U^2_{12}$.}
  \label{fig:Usets}
\end{subfigure}
\caption{We evaluate the action $S_\lattice$ defined in \eqref{eq:latticeaction} on a flat connection described by local connection 1-forms: $\connection_1=0$ defined on $U_1=\{(x,y,z):y<|x|, 0 \leq z\leq \ell\}$ and $\connection_2=0$ defined on $U_2=\{(x,y,z):y>-|x|, 0 \leq z\leq \ell\}$. The only non-trivial contribution to the action comes from the holonomies around loops encircling the defect $U_{p;12}$, where $p=\latticeconst(-\frac{1}{2},-\frac{1}{2},c), c\in\mathbb{Z}$. This holonomy is equal to a product of components of the locally constant transition map $g_{12}$ taking value $g_{12}^1$ on $U_{12}^1$ and $g_{12}^2$ on $U_{12}^2$: $U_{p;12}=g_{12}^1\ g_{21}^2$. The matrices $g_{12}^1$ and $g_{12}^2$ correspond to discretized connection in the tetrad formulation of Regge calculus.}
\label{fig:latticedef}
\end{figure}
We define the following functional 
\begin{equation}\label{eq:latticeaction}
S_\lattice[\triad,\connection]=\sum_{p\in \lattice} 2 \left(\Tr{\triad(\partial_1)|_p \curvaturematrix_{p;23}}+\Tr{\triad(\partial_2)|_p \curvaturematrix_{p;31}}+\Tr{\triad(\partial_3)|_p \curvaturematrix_{p;12}}\right) \latticeconst^3.
\end{equation}
Clearly, it is gauge invariant. If the fields $\triad$ and $\connection$ are defined by a restriction to $\manifold$ of global 1-forms on $\overline{\manifold}$ (in particular, if the SO($3$) principal fibre bundle over $\manifold$ is trivial), then in the limit of lattice constant going to $0$ we recover the BF action:
$$
S[\triad,\connection]:=\lim_{\latticeconst\to 0} S_\lattice[\triad,\connection]=\int_{\overline{\manifold}} \eta_{ij}\triad^i\wedge \curvature^j.
$$
Let us consider a flat connection (not necessarily torsionless) and a triad $\triad$ on $\mathbb{R}^2_*\times [0,\ell]$. We introduce a covering of $\mathbb{R}^2_*\times [0,\ell]$ with two open sets:
$$
U_1=\{(x,y,z):y<|x|, 0 \leq z\leq \ell\}, \quad U_2=\{(x,y,z):y>-|x|, 0 \leq z\leq \ell\}
$$
The intersection $U_1\cap U_2$ is formed by two disjoint components (see figure \ref{fig:Usets}):
$$
U_{12}^1=\{(x,y,z):x<-|y|, 0 \leq z\leq \ell\},\quad U_{12}^2=\{(x,y,z):x>|y|, 0 \leq z\leq \ell\},
$$
Clearly $S_\lattice[\triad,\connection]$ is gauge invariant. Let us fix the gauge such that the local connection 1-forms $\connection_1$ and $\connection_2$ defined on $U_1$ and $U_2$, respectively, vanish:
$$\connection_1=\connection_2=0.$$ In this gauge the transition function $g_{12}$ is locally constant (see for example \cite{Thurston}):
 $$
g_{12}(x)=\begin{cases}g_{12}^1\in {\rm SO(}n{\rm )\ if\ }x\in U^1_{12} ,\\ g_{12}^2\in {\rm SO(}n{\rm )\ if\ }x\in U^2_{12}.\end{cases}
 $$
A triad $e$ defines and is defined by local 1-forms $\triad_1$ and $\triad_2$ defined on $U_1$ and $U_2$, respectively. The so(3)-valued matrices $\curvaturematrix_{p;\mu\nu}$ defined by holonomies $U_{p;\mu\nu}$ by relation \eqref{eq:Fdef} have the following properties:
$$
\epsilon^2 \curvaturematrix_{p;\mu\nu}=\begin{cases}\curvaturematrix {\rm\ if\ }p=\latticeconst(-\frac{1}{2},-\frac{1}{2},c), c\in\mathbb{Z} {\rm\ and\ }\mu=1,\nu=2,\\
-\curvaturematrix {\rm\ if\ }p=\latticeconst(-\frac{1}{2},-\frac{1}{2},c), c\in\mathbb{Z} {\rm\ and\ }\mu=2,\nu=1,
 \\ 0{\rm\ otherwise},\\ \end{cases}
$$
where $\curvaturematrix$ is given by:
$$
\exp(\curvaturematrix)=g_{12}^1\ g_{21}^2.
$$
Due to these properties the only non-zero contributions to $S_\lattice[\triad,\connection]$ come from points $p$ in the set $\edge_\epsilon=\{\latticeconst(-\frac{1}{2},-\frac{1}{2},c): c\in\{0,1,\ldots,\frac{\ell}{\latticeconst}\}\}$. The limit $\lim_{\latticeconst\to 0} S_\lattice[\triad,\connection]$ becomes
$$
\lim_{\latticeconst\to 0} S_\lattice[\triad,\connection]=\lim_{\latticeconst\to 0}\Sigma_{p\in\edge_{\epsilon}}\Tr{\triad_1(\partial_z)|_p\, \curvaturematrix} \latticeconst=\Tr{\triad_\edge\, \curvaturematrix},
$$
where $\triad_\edge$ is the integral of $\triad$ over the defect
$$
\triad_\edge:=\lim_{\latticeconst\to 0}\int_0^\ell\ \triad_1(\partial_z)|_{p=(-\frac{1}{2}\latticeconst,-\frac{1}{2}\latticeconst,t)}\, dt.
$$
The action functional is similar to a part of the action functional in the tetrad formulation of the Regge calculus \cite{Khatsymovsky, Bander}. The role of discretized triad is played by the vector $\triad_\edge$ and the role of discretized connection is played by (locally constant) transfer matrices $g^1_{12}, g^2_{12}$ of a flat bundle.

\section{The space-time as a simplicial complex}
In this section we describe the space-time as certain simplicial complex equipped with a geometric structure. 

An \textbf{abstract simplicial complex} $\complex$ is a collection of non-empty finite sets called abstract simplices such that every non-empty subset of an abstract simplex in $\complex$ is an abstract simplex in $\complex$\cite{Pontryagin,LeeTopMan}. Every non-empty subset of an abstract simplex $\simplex$ is called a face of $\simplex$. We will write $\simplex' \isfaceof \simplex$ if $\simplex'$ is face of $\simplex$. In the following we will restrict to finite complexes, i.e. to complexes that are finite sets. Any element of an abstract simplex $\simplex$ is called a vertex of $\simplex$. A simplex $\Delta$ with $r+1$ vertices is said to have dimension $r$ and will be called an $r$-simplex. The maximum of the dimensions of simplices contained in $\complex$  is called the dimension of the complex $\complex$ and will be denoted by $\dim \complex$. We will denote by $\complex^{(r)}, r\in\{0,\ldots,\dim \complex\}$ the set of all $r$-simplices in $\complex$. Clearly $\complex=\bigcup_{r=0}^{\dim \complex} \complex^{(r)}$.

A \textbf{pseudo-manifold} \cite{Encyclopediaofmath} of dimension $n$ is a finite $n$-dimensional simplicial complex with the following properties:
\begin{itemize}
\item each (n-1)-dimensional simplex is a face of precisely two $n$-dimensional simplices,
\item any two $n$-simplices $\simplex$ and $\simplex'$ can be connected by a chain of simplices $\simplex_0=\simplex,\ldots,\simplex_k=\simplex'$ such that $\simplex_i\cap \simplex_{i+1}$ is an $(n-1)$-dimensional simplex,
\item each simplex is a face of some $n$-simplex.
\end{itemize}

We will further assume that the pseudomanifold is orientable, i.e. $H^{n}(\complex)=\mathbb{Z}$ \cite{Spanier}. We will denote by $\vol[\complex]$ a representative of the generator of this homology group. 

We will now introduce additional structure on a pseudo-manifold, which we will call a geometric structure. This structure is equivalent to assigning lengths to edges in the Regge calculus and will carry the information about the orientation. Let us first recall the definition of a (geometric) $n$-simplex \cite{Pontryagin}. A set of points $\{x_0,x_1,\ldots,x_k\}$ in $\mathbb{R}^n$ is called (affinely) independent if the system of vectors
$$
(x_1-x_0),\ (x_2-x_0),\ \ldots,\ (x_k-x_0)
$$
is linearly independent. A convex hull of a set of independent points $\{x_0,x_1,\ldots,x_k\}$ in $\mathbb{R}^n$ is called a (geometric) $k$-simplex in $\mathbb{R}^n$. Let us denote by $\simplex_1,\ldots, \simplex_N$ the $n$-dimensional simplices of $\complex$. Consider a family of maps $\{\sigma_i\}_{i\in\{1,2,\ldots,N\}}$, $\sigma_i:\simplex_i\to \mathbb{R}^n$ assigning a set of $n+1$ independent points in $\mathbb{R}^n$ to the set of vertices of a simplex $\simplex_i$. Each map $\sigma_i$ will be called a geometric realization of the (abstract) simplex $\simplex_i$. We will use the following notation:
\begin{eqnarray*}
\gsimplex:=\conv(\sigma_i(\simplex_i)),\\
|\simplex|_i:=\conv(\sigma_i(\simplex)){\rm\ for\ any\ } \simplex\subset \simplex_i,
\end{eqnarray*}
where $\conv$ denotes the convex hall of a set of points. We will call $|\simplex_i|,|\simplex|_i$ geometric simplices. By $b$ we will denote a barycenter of a (geometric) simplex:
$$
b(|\simplex_i|):=\frac{1}{n+1} \sum_{v\in \simplex_i} \sigma_i(v), \quad b(|\simplex|_{i}):=\frac{1}{\card\simplex} \sum_{v\in \simplex} \sigma_i(v).
$$ We will assume that each geometric simplex $\gsimplex$ has an orientation that agrees with the standard orientation of $\mathbb{R}^n$, i.e.
\begin{equation*}\label{eq:orientability}
\vol[\kappa]([v_0,v_1,\ldots,v_n])=\sgn\left(\epsilon_{I_1 I_2 \ldots I_n} (\sigma_i(v_1)-\sigma_i(v_0))^{I_1}\ldots (\sigma_i(v_N)-\sigma_i(v_0))^{I_n}\right),
\end{equation*}
where $\epsilon_{I_1 I_2 \ldots I_n}$ is the Levi-Civita symbol such that $\epsilon_{0 1 \ldots n-1}=1$. We say that two maps $\sigma_i$ and $\sigma_j$ are compatible if one of the following conditions is satisfied:
\begin{itemize}
\item $\dim (\simplex_i\cap\simplex_j)<n-1$,
\item there is an affine isometry $\alpha_{ij}$ of the Euclidean\slash Minkowski space $(\mathbb{R}^n,\intmetric)$ preserving the orientation such that
\begin{equation}\label{eq:gluingpoints}
\alpha_{ij}(\sigma_j(v))=\sigma_i(v){\rm\ for\ any\ }v\in \simplex_i\cap\simplex_j.
\end{equation}
\end{itemize}
\begin{rmk}\label{rmk:oppositesides}
Maps $\alpha_{ij}$ are unique. The condition \eqref{eq:gluingpoints} specifies an affine isometry uniquely up to reflection about the hyperplane $\hyperplane$ defined by points $\sigma_i(v), v\in\simplex_i\cap \simplex_j$. Let $v_0$ be the unique element in $\simplex_i\backslash\simplex_j$ and $w_0$ be the unique element in $\simplex_j\backslash\simplex_i$. The points $\sigma_i(v_{0})$ and $\alpha_{ij}(\sigma_j(w_0))$ are on opposite sides of $\hyperplane$ because the complex is orientable, the orientation of each geometric simplex $\gsimplex$ has an orientation that agrees with the standard orientation of $\mathbb{R}^n$ and $\alpha_{ij}$ is orientation preserving. This fixes the remaining ambiguity.
\end{rmk} 
A \textbf{geometric structure} on an $n$-dimensional pseudo-manifold $\complex$ is a family of pairwise compatible maps $\{\sigma_i\}_{i\in\{1,2,\ldots,N\}}$.  The maps $\alpha_{ij}$ have the following properties:
$$
\alpha_{ii}=\id, \quad \alpha_{ij}=\alpha_{ji}^{-1}. $$ We will call it a \textbf{gluing pattern}. We say that two geometric structures $\{\sigma_i\}_{i\in\{1,2,\ldots,N\}}$ and $\{\sigma_i'\}_{i\in\{1,2,\ldots,N\}}$ are equivalent if there exist affine isometries $r_i:\mathbb{R}^n\to \mathbb{R}^n$ such that
$$
\alpha'_{ij}=r_i^{-1} \alpha_{ij} r_j.
$$

\begin{ex}
Let $\complex$ be a $n$-dimensional pseudomanifold with Euclidean geometric structure $\sigma$. To any 1-simplex $e=\{v_0,v_1\}\in \complex^{(1)}$ we assign its length
$$
\length_e=\intmetric(\sigma_i(v_1) -\sigma_i(v_0),\sigma_i(v_1) -\sigma_i(v_0)),
$$
where $\simplex_i$ is any $n$-simplex which face is $e$. This way we obtain Regge variables corresponding to the geometric structure. Let us note that equivalent geometric structures lead to the same Regge variables. Since each geometric $n$-simplex is defined uniquely up to an affine isometry by the lengths of its edges, there is 1-1 correspondence between Regge variables and equivalence classes of geometric structures. 
\end{ex}

\begin{prop}\label{prop:canonicalgeometricstructure}Each pseudomanifold can be equipped with a geometric structure.\end{prop} \begin{proof}We construct an Euclidean geometric structure on a pseudomanifold $\complex$ in the following way. Each map $\sigma_i:\simplex_i\to \mathbb{R}^n$ assigns to the vertices of the simplex $\simplex_i$ the points of the standard regular $n$-simplex \footnote{The standard regular $n$-simplex is an $n$-simplex in $\mathbb{R}^n$ which edges have unit length  \cite{Regularnsimplex}.} in such a way that the orientability condition \eqref{eq:orientability} is satisfied. This orientability condition and congruence of the faces of regular simplices guarantees that any pair of maps $\sigma_i$, $\sigma_j$ is compatible. \end{proof}

\section{The space-time as a manifold with defects}
With a pseudo-manifold $\complex$ equipped with a geometric structure we can associate the following structure:
\begin{itemize}
\item a finite set of (geometric) $n$-simplices;
\item a choice of pairs of $(n-1)$-dimensional faces of the geometric $n$-simplices such that each face appears in precisely one of the pairs: by definition each $(n-1)$-simplex is a face of precisely two $n$-simplices; denote by $\simplex_{ij}$ the $(n-1)$-simplex that is a face of $n$-simplices $\simplex_i$ and $\simplex_j$; it defines a pair of geometric $(n-1)$-simplices $|\simplex_{ij}|_i$ and $|\simplex_{ij}|_j$ that are faces of the geometric simplices $\gsimplex$ and $\gsimplex[j]$;
\item an affine identification map between the faces of each pair: such identification is given by the gluing pattern $\{\alpha_{ij}\}$.
\end{itemize}
Such structure is called a gluing \cite{Thurston}.  Each of the geometric $n$-simplices is equipped with topology induced from the standard topology on $\mathbb{R}^n$. The quotient space of the topological sum of all the $n$-simplices by the equivalence relation generated by the gluing pattern $\{\alpha_{ij}\}$ will be also called a gluing and denoted by $\gluing$ \footnote{For the definitions of the induced topology, topological sum and the quotient space we refer the reader to \cite{Spanier}.}. We remove from this space a closed subset formed by simplices of dimension not exceeding $n-2$. The resulting topological space will be denoted by $\manifold$. It can be equipped with a smooth atlas making it a smooth manifold, which will be called a \textbf{manifold with defects}. We will now construct a coordinate system. For each pair $(\simplex_i, \simplex_{ij})$ of an $n$-simplex $\simplex_i$ and its $(n-1)$-dimensional face $\simplex_{ij}=\simplex_i\cap \simplex_j$ we construct a closed subset of $\gsimplex$ denoted by $\overline{U}_{ij}$ that is a convex hull of $\sigma_i(\simplex_{ij})$ and the barycenter of $\gsimplex$. We will also denote by $\gsimplex$ and $\overline{U}_{ij}$ the canonical inclusions of $\gsimplex$ and $\overline{U}_{ij}$ into the gluing $\gluing$. The open covering of the manifold $\manifold$ is of the following form:
$$
U_{i}=\interior \left( \gsimplex\cup \bigcup_{\simplex_{ij}\isfaceof \simplex_i} \overline{U}_{ij} \right),
$$
where $\interior$ denotes the interior of a subset of $\gluing$. The coordinate charts
$$
\phi_{i}:U_i\to \mathbb{R}^n
$$
 are homeomorphisms from $U_i$ to open subsets of $\mathbb{R}^n$ defined in the following way: 
\begin{itemize}
\item We consider a continuous map $$\tilde{\phi}_{i}: \gsimplex\vee \bigvee_{j:\, \dim(\simplex_i\cap \simplex_j)=n-1} \gsimplex[j] \to \mathbb{R}^4 $$
defined by the following properties:
$$
\tilde{\phi}_{i}((x,i))=x, \quad \tilde{\phi}_{i}((x,j))=\alpha_{ji}^{-1}(x) {\rm\ for\ }j\neq i.
$$
\item Clearly it is continuous and has the property that $\tilde{\phi}_{i}((x,i))=\tilde{\phi}_{i}((y,j))$ whenever $(x,i)\sim (y,j)$. Therefore this map descents to a continuous map on the quotient $\gsimplex\vee \bigvee_{j:\, \dim(\simplex_i\cap \simplex_j)=n-1}\gsimplex[j]\slash\sim$. The resulting map restricted to $U_i$ is the coordinate chart $\phi_{i}: U_i \to \mathbb{R}^n$. Indeed, it is continuous. Thanks to the fact that $\alpha_{ij}$ is orientation preserving and satisfies property \eqref{eq:gluingpoints} it is also 1-1 (see also Remark \ref{rmk:oppositesides}).
\end{itemize}
Let $(U_i,\phi_i)$ and $(U_j,\phi_j)$ be two coordinate charts such that $U_i\cap U_j\neq \emptyset$. It is easy to check that the transition functions $\phi_{ji}: \phi_i(U_i\cap U_j)\to \phi_j(U_i\cap U_j),\ \phi_{ji}=\phi_j \phi_i^{-1}$ coincide with the gluing pattern:
$$
\phi_{ji}=\alpha_{ji}|_{\phi_i(U_i\cap U_j)}.
$$
Since $\alpha_{ij}$ are affine isometries of $\mathbb{R}^n$ the transition functions are clearly smooth. Therefore $\manifold$ is a smooth manifold. Let us note that this shows also that $\manifold$ is also a piecewise linear manifold and if the metric $\intmetric$ is Euclidean it is also an Euclidean manifold (for the definitions we refer the reader to \cite{Thurston}).

Let $\sigma_i$ and $\mu_i$ be two geometric structures on a pseudomanifold $\complex$. The corresponding gluings are homeomorphic and will be denoted by $\gluing$. Denote by $(\phi_i,U_i)$, $(\varphi_i,V_i)$ the corresponding coordinate charts on $\gluing$.
\begin{prop}The charts $(\phi_i,U_i)$ and $(\varphi_j,V_j)$ are compatible.\end{prop}
\begin{proof}
Let us note that for each $i$ there exists an affine map $\beta_i:\mathbb{R}^n\to \mathbb{R}^n$ such that 
$$
\mu_i=\beta_i \circ \sigma_i.
$$
The transition functions $\tilde{\phi}_{ji}: \phi_i(U_i\cap V_j)\to \varphi_j(U_i\cap V_j),\ \phi_{ji}=\varphi_j \phi_i^{-1}$ are given by:
$$
\tilde{\phi}_{ji}=\beta_j\circ\alpha_{ji}|_{\phi_i(U_i\cap V_j)}.
$$
Clearly, they are smooth diffeomorphism and therefore the charts are smoothly compatible.
\end{proof}
As a result different geometric structures on a pseudomanifold correspond to different charts in the same smooth structure (i.e. maximal smooth atlas). Let us recall that Proposition \ref{prop:canonicalgeometricstructure} guarantees that each pseudomanifold can be equipped with a geometric structure. As a result to each pseudomanifold there corresponds a unique manifold with defects.
\begin{ex}\label{ex:manifold}
Let $S=\{v_0,v_1,v_2,v_3\}$ be a 4-element set. We consider a simplicial complex $\complex$ that is the set of all proper subsets of $S$. In fact, this is a pseudo-manifold, called also a simplicial sphere. We equip the pseudo-manifold with the geometric structure constructed in the proof of Proposition \ref{prop:canonicalgeometricstructure}: to each 3-element subset $\simplex_{i}, i\in\{0,1,2,3\}$ of $S$ such that $\simplex_i\cup \{v_i\}=S$ we assign a map $\sigma_i:\simplex_i\to \mathbb{R}^2$ defined by the following formula:
$$
(\sigma_i(v_0),\ldots,\sigma_i(v_{i-1}),\sigma_i(v_{i+1}),\ldots,\sigma_i(v_3))=\begin{cases}(p_0,p_1,p_2),\textrm{ if }i\textrm{ is even},\\ (p_2,p_1,p_0),\textrm{ if }i\textrm{ is odd},\end{cases}
$$
where
$$
p_0=(-\frac{1}{2},-\frac{\sqrt{3}}{6}),\quad p_1=(\frac{1}{2},-\frac{\sqrt{3}}{6}),\quad p_2=(0,\frac{\sqrt{3}}{3}).
$$
Clearly the geometric 2-simplices $\gsimplex$ are equilateral triangles. The gluing pattern is
\begin{align*}
\alpha_{01}(x)=\rotate{\frac{\pi}{3}}(x)(x-p_1)+p_1, \quad \alpha_{02}(x)=\rotate{\frac{\pi}{3}}(x-p_2)+p_2,\\ \alpha_{03}(x)=\rotate{\pi}(x-p_0)+p_1,\quad
\alpha_{12}(x)=\rotate{\pi}(x-p_0)+p_2,\\ \alpha_{13}(x)=\rotate{\frac{\pi}{3}}(x-p_2)+p_2,\quad \alpha_{23}(x)=\rotate{-\frac{\pi}{3}}(x-p_1)+p_1,
\end{align*}
where $\rotate{\alpha}$ is a clockwise rotation around $(0,0)$ through the angle $\alpha$. The gluing is homeomorphic to $\mathbb{S}^2$ (a boundary of a tetrahedron) and the manifold with defects is diffeomorphic with a sphere with 4 punctures -- see figure \ref{fig:punctured_sphere}.
\begin{figure}
\begin{subfigure}{0.5\textwidth}
 \centering
\includegraphics[width=.8\linewidth]{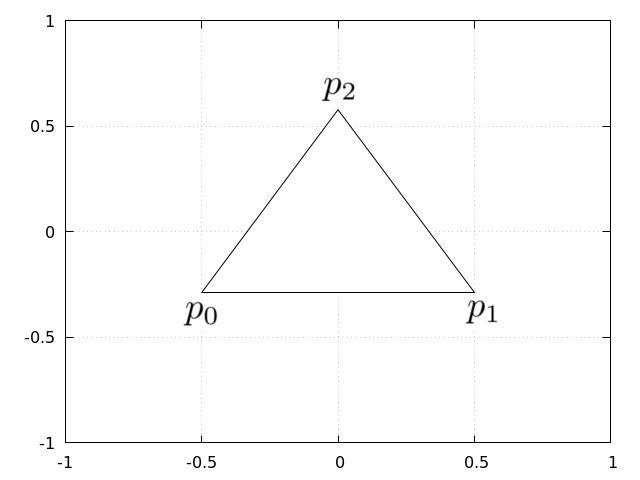}
  \caption{Each geometric 2-simplex $\gsimplex$ in Example \ref{ex:manifold} is an equilateral triangle.}
  \label{fig:triangle}
\end{subfigure}
\begin{subfigure}{0.5\textwidth}
 \centering
\includegraphics[width=.8\linewidth]{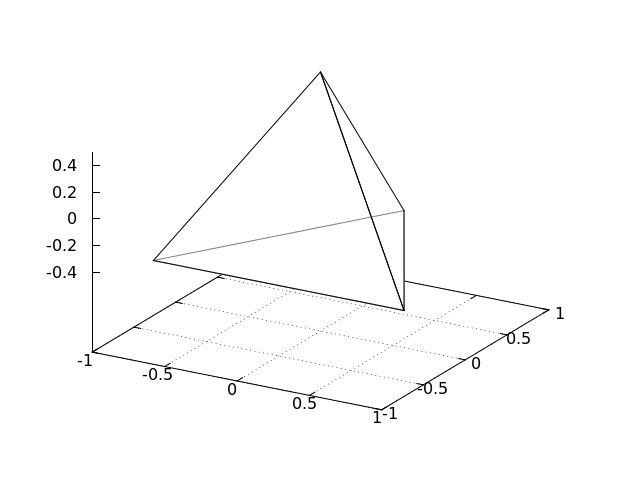}
  \caption{The gluing $\gluing$ is homeomorphic to a boundary of a tetrahedron, i.e. a 2-sphere.}
  \label{fig:gluing}
\end{subfigure}
\begin{subfigure}{0.5\textwidth}
 \centering
\includegraphics[width=.8\linewidth]{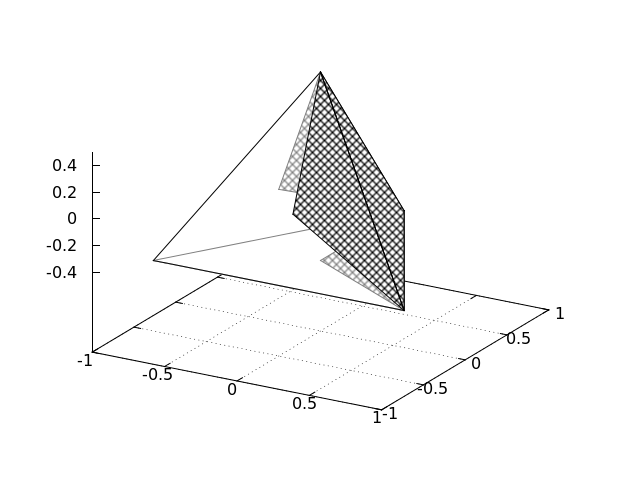}
  \caption{Each open set $U_i$ is an interior in $\gluing$ of a set $\gsimplex\cup \bigcup_{\simplex_{ij}\isfaceof \simplex_i} \overline{U}_{ij}$.}
  \label{fig:tetrahedron_open_set}
\end{subfigure}
\begin{subfigure}{0.5\textwidth}
 \centering
\includegraphics[width=.8\linewidth]{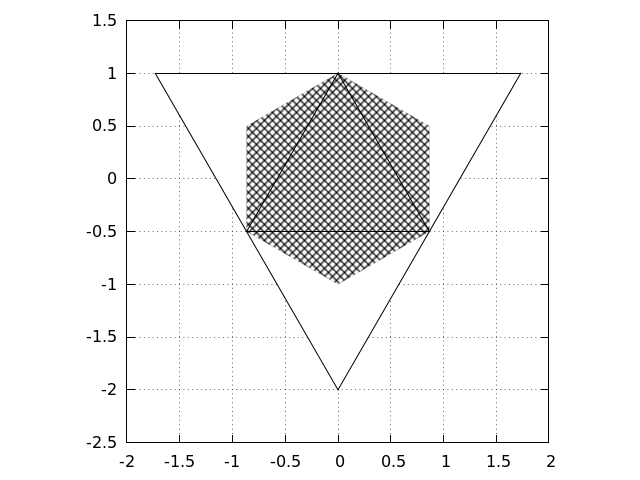}
  \caption{The coordinate map $\phi_i$ maps the open set $U_i$ to an a set $\phi_i(U_i)$. The set $\phi_i(U_i)$ is denoted on the figure by the shaded area. It is a subset of a set $\tilde{\phi}_i(\gsimplex\vee \bigvee_{j:\, \dim(\simplex_i\cap \simplex_j)=n-1}\gsimplex[j]\slash\sim)$ formed by 4 triangles.}
  \label{fig:open_set}
  \end{subfigure}
\caption{Illustration to Example \ref{ex:manifold}.}
\label{fig:exgluing}
\end{figure}
\begin{figure}
\centering
\includegraphics[width=.3\textwidth]{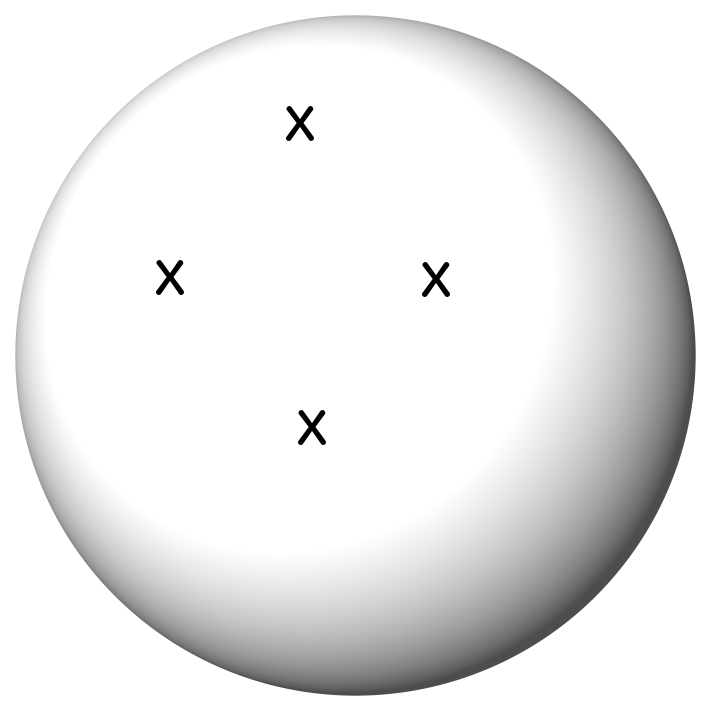}
\caption{In Example \ref{ex:manifold} the manifold with defects $\manifold$ is diffeomorphic with a sphere with 4 punctures.}
\label{fig:punctured_sphere}
\end{figure}
\end{ex}
\begin{rmk}
A gluing in general is not a topological manifold because there may be points which do not have any neighbourhood homeomorphic to a ball. Interesting example of such gluing is Example 1.4.8 from \cite{Thurston}. However such points may lay only inside simplices of dimension not exceeding $n-2$. Since we remove such simplices from the gluing, we always obtain a manifold. 
\end{rmk}
\begin{rmk}
A gluing is more general concept than a pseudo-manifold with geometric structure. The non-trivial requirement is that an intersection of two simplices has to be a simplex. However, our construction of a manifold can be easily generalized to gluings. In this generalization:
\begin{itemize} 
\item The map $\tilde{\phi}$ should be constructed for sets $\gsimplex\vee \bigvee_{\simplex_{ij}\isfaceof \simplex_i} \overline{U}_{ij}$.
\item The intersection of two open sets $U_i\cap U_j$ could be disconnected. In such case the transition function is locally constant and coincides with different affine identification maps on different connected components.
\end{itemize}
\end{rmk}

There is a canonical metric on $\manifold$. In the coordinate system constructed above it coincides with the metric $\intmetric$ on $\mathbb{R}^n$:
$$
\metric_{\mu\nu}=\intmetric_{\mu\nu}.
$$
Indeed, since in this coordinate system the transition functions are SO($n$) transformations and $\intmetric$ is SO($n$) invariant, the family of tensors $\intmetric|_{\phi_i(U_i)}$ defines a tensor field on $\manifold$. In the following we will use the same symbol $\intmetric$ to denote the metric on $\mathbb{R}^n$ and the canonical metric on $\manifold$.

\section{BF theory on manifolds with defects}
\subsection{BF theory}
BF theory is a gauge theory, which gauge group can be arbitrary group such that its Lie algebra is equipped with and invariant nondegenerate bilinear form. In this paper we will be interested only in the SO($\intmetric$) group, i.e. the group of matrices acting in $\mathbb{R}^n$, having unit determinant and preserving a bilinear form $\intmetric$. We consider an SO($\intmetric$) principal fibre bundle $\bundle$ over $n$-dimensional oriented smooth manifold $\spacetime$ modeling the space-time and the fibre bundle $\abundle$ associated to $\bundle$ via the adjoint action of SO($\intmetric$) on its Lie algebra. The field variables are: a connection $\connection$ on $\bundle$ and an $\abundle$-valued $(n-2)$-form $\B$.  We will denote by $\curvature$ the curvature of the connection $\connection$. Let us assume for a moment that $\manifold$ is compact. The action for the theory is:
\begin{equation}\label{eq:BFactioncompact}
S[\B,\connection]=\int_\spacetime\, \Tr{\B\wedge \curvature}.
\end{equation}
The theory of Gravity can be treated as a constrained BF theory \cite{Plebanski}. In this case $\dim \manifold=4$ and $\intmetric=\diag(-1,1,1,1)$ is the Minkowski metric. Consider a fibre bundle $\dbundle$ associated to $\bundle$ via the defining representation $\rho$. Let $e^I$ be a vierbein, i.e. $\dbundle$-valued $1$-form. The constraints are restricting the $\B$ fields to be of the form
$$
\B^{IJ}=\star \left( e^I \wedge e^J \right)
$$
for some vierbein $e^I$. The star $\star$ denotes the Hodge star and the internal indices denoted by capital Latin letters $I,J,K,L,\ldots$ are lowered and raised with the metric $\intmetric$.
\subsection{BF theory on manifolds with defects}\label{sc:action_functional}
There is a straightforward generalization of the action functional \eqref{eq:BFactioncompact} to our manifolds with defects. One could simply define the action functional to be an integral over the manifold with defects of the form $\Tr{\B\wedge \curvature}$. However such definition does not deal properly with field configurations such that the curvature $\curvature$ has distributional support concentrated on the defect. Since they are crucial in the Regge calculus, we will propose a regularization of the action functional that includes such configurations. Our regularization is based on appropriate discretization of the action functional. We will start with a subdivision of the manifold with defects and construct the action functional as a refinement limit of action functionals defined on subdivisions.

\subsubsection{Subdivisions of a pseudomanifold}
Let $(\complex,\sigma)$ and $(\complex',\sigma')$ be two pseudo-manifolds of dimension $n$ with geometric structures. We will say that $\varphi:\gluing[\complex']\to \gluing$ is \textbf{linear} if the following condition is satisfied: if $\simplex'_{i'}=\{v_0,\ldots,v_n\}$ is an $n$-simplex in $\complex$ and $x=\lambda_0 \sigma'_{i'}(v_0)+\ldots+\lambda_{n} \sigma'_{i'}(v_n)\in \gsimplex[i']'$ then the points $\varphi(\sigma'_{i'}(v_0)),\ldots, \varphi(\sigma'_{i'}(v_n))$ belong to some simplex $\gsimplex$ in $\gluing$ and $\varphi(x)=\lambda_0 \varphi(\sigma'_{i'}(v_0)) + \ldots +\lambda_n \varphi(\sigma'_{i'}(v_n))$. Let $\complex$ and $\complex'$ be two pseudomanifolds and let $\mathring{\sigma},\mathring{\sigma}'$ be the corresponding canonical geometric structures constructed in Proposition \ref{prop:canonicalgeometricstructure}. Let us denote by $\gluing[\complex']_0$ and $\gluing_0$ the corresponding gluings. We will say that $\complex'$  is a \textbf{subdivision} of $\complex$ if there is a linear homeomorphism $\varphi: \gluing[\complex']_0\to\gluing_0$. 

A \textbf{diameter} of an $n$-simplex $\simplex'_{i'}\in {\complex'}^{(n)}$ in a subdivision $\complex'$ of $\complex$ will be defined by
$$
\diam(\simplex'_{i'})=\sup\{d(x,y): x,y\in \varphi(|\simplex_{i'}'|)\},
$$  
where $d$ is the standard Euclidean metric in $\mathbb{R}^n$. Let us note that the diameter of an $n$-simplex in the subdivision $\complex'$ is calculated with respect to the metric on $\complex$. In this sense we use on $\complex'$ a metric induced from $\complex$. A \textbf{mesh} of a subdivision $\complex'$ of $\complex$ will be defined by:
$$
\mesh(\complex')=\sup\{\diam(\simplex_{i'}'): \simplex_{i'}'\in {\complex'}^{(n)} \}.
$$ 
A mesh of a subdivision can be arbitrary small \cite{Spanier}. We will say that a sequence of subdivisions $\complex'_m$ is a \textbf{regular refinement} if it satisfies the following conditions: 
\begin{itemize}
\item
$$
\lim_{m\to \infty}\mesh(\complex'_m)=0,
$$
\item
in each $\complex'_m$ every $(n-2)$-simplex $\simplex_f$ is shared by at most $N_{\rm max}$ $n$-simplices, where the number $N_{\rm max}$ does not depend on $m$,
\item the number $N_{\simplex,m}$ of $k$-simplices in $\complex'_m$ subdividing a $k$-simplex $\simplex$ satisfies:
$$N_{\simplex,m} =\bigo(\epsilon^{-k}_m),$$
where $\epsilon_m=\mesh(\complex'_m)$.
\end{itemize}
\begin{figure}
\begin{subfigure}{0.45\textwidth}
 \centering
\includegraphics[width=.9\linewidth]{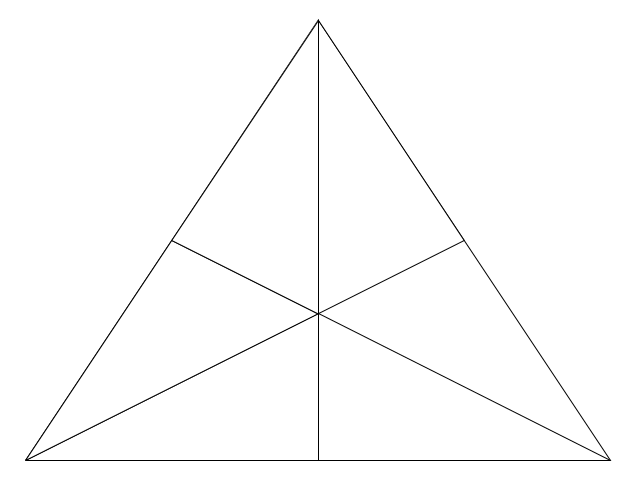}
  \caption{Barycentric subdivision of a triangle}
  \label{fig:barycentric}
\end{subfigure}
\begin{subfigure}{0.45\textwidth}
 \centering
\includegraphics[width=.9\linewidth]{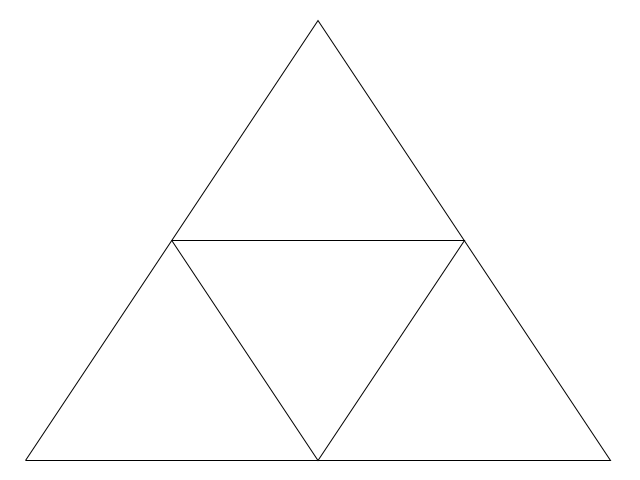}
  \caption{Edgewise subdivision ${\rm Esd}_2$ of a triangle. In two dimensions considered first by Freudenthal \cite{Freudenthal}. Generalized to any dimension in \cite{EdgeSubdivision}.}
  \label{fig:edge_subdivision}
\end{subfigure}
\caption{We consider only sequences of subdivisions $\complex'_m$ satisfying certain regularity conditions. In particular the (iterated) barycentric subdivision (a) does not satisfy the requirement, because we assume that in each $\complex'_m$ every $(n-2)$-simplex $\simplex_f$ is shared by at most $N_{\rm max}$ $n$-simplices, where the number $N_{\rm max}$ does not depend on $m$. An example of a subdivision satisfying these conditions is the edgewise subdivision from \cite{EdgeSubdivision} illustrated on (b).}
\label{fig:regularsubdivisions}
\end{figure}

The first condition on the sequence of subdivisions is obvious. The next two requirements will be used in Section \ref{sc:ELBFdefects}. They exclude for example the barycentric subdivision \cite{Spanier} (see figure \ref{fig:barycentric}). An example of a regular refinement is the following. The complex $\complex'_m$ is obtained by performing an edgewise subdivision ${\rm Esd}_m$ from \cite{EdgeSubdivision} of each $n$-simplex (see for example figure \ref{fig:edge_subdivision}). By the main theorem of \cite{EdgeSubdivision} we know that each $n$-simplex is divided into $m^n$ $n$-simplices of equal volume which fall into at most $\frac{n!}{2}$ congruence classes. Denote by $\simplex_1,\ldots,\simplex_M$ the $n$-simplices of $\complex$ and by $|\simplex_i^j|, j\in\{1,\ldots, \frac{n!}{2}\}$ representatives of the congruence classes of possible $n$-simplices appearing in the subdivisions of the $n$-simplices $|\simplex_i|$ scaled such that the volume of $|\simplex_i^j|$ is equal to the volume of $|\simplex_i|$. Let 
$$
{\rm Max}(\complex)=\max\{\diam(|\simplex^j_i|): i\in\{1,\ldots,M\}, j\in \{1,\ldots, \frac{n!}{2}\}\},
$$  
$$
{\rm Min}(\complex)=\min\{\diam(|\simplex^j_i|): i\in\{1,\ldots,M\}, j\in \{1,\ldots, \frac{n!}{2}\}\}.
$$  
Using this notation we have: 
$$
{\rm Min}(\complex)\frac{1}{m}\leq\mesh(\complex'_m)\leq {\rm Max}(\complex) \frac{1}{m}.
$$
This shows that $\mesh(\complex'_m)=\bigo(m^{-1})$ and therefore the first property is shown. We will now show the second property. Denote by $\dihedral_f$ the sum of dihedral angles around the $(n-2)$-simplex $\simplex_f$. Let
$$\theta_{\rm max}=\max\{\{2\pi\}\cup\{\dihedral_f:\simplex_f\in \complex^{(2)}\}\}.$$ Let us notice that the sum of dihedral angles around an $(n-2)$-simplex $\simplex_{f'}$ in $\complex'_m$ (in the Euclidean metric induced on the subdivision) is smaller or equal $\theta_{\rm max}$:
$$
\dihedral_{f'}\leq \theta_{\rm max}.
$$ 
Let $\alpha_{\rm min}$ be the minimal dihedral angle that can appear in the $n$-simplices $|\simplex^j_i|$. The upper bound $N_{\rm max}$ from the second point can be chosen to be the smallest integer greater or equal {$\theta_{\rm max}\slash \alpha_{\rm min}$}:
$$N_{\rm max}=\left\lceil\frac{\theta_{\rm max}}{\alpha_{\rm min}}\right\rceil.$$ What is left is to check the third point. Another property of the Edge Subdivision of a simplex is that its faces are subdivided in the same way. This in particular means that the number of $k$-simplices subdividing a $k$-simplex $\simplex$ is $m^{k}$:
$$
N_{\simplex,m}= m^{k}.
$$ 
We have shown above that $m=\bigo(\epsilon^{-1}_m)$. Therefore
$$
N_{\simplex,m}= \bigo(\epsilon^{-k}_m).
$$ 

\subsubsection{The action functional}
Let $\complex$ be a pseudomanifold of dimension $n$ and $\manifold$ be the corresponding manifold with defects. We consider a sequence of $n$-simplices $f=(\simplex_1,\ldots,\simplex_N)$ satisfying the following properties:
\begin{itemize}
\item $\simplex_f:=\bigcap_{i=1}^N \simplex_i$ is an $(n-2)$-simplex,
\item $\simplex_k\cap\simplex_{k+1}$ is an $(n-1)$-simplex for $k\in\mathbb{Z}_N$.
\end{itemize} 
We will call $f$ a face dual to $\simplex_f$ or simply a dual face. For each $(n-2)$-simplex there exists a dual face, because each $(n-1)$-simplex is a face of precisely two different $n$-simplices, in a given simplex $\simplex_k$ the simplex $\simplex_f$ is a face of precisely two different $(n-1)$-simplices, the simplicial complex $\complex'$ is connected and finite. Let us note that whenever $f=(\simplex_1,\ldots,\simplex_N)$ is a face dual to an $(n-2)$-simplex, then so is $f^{-1}:=(\simplex_1,\simplex_N,\simplex_{N-1},\ldots,\simplex_2)$.  In fact, any cyclic permutation of the $n$-simplices in $f$ or $f^{-1}$ is a face dual to the $(n-2)$-simplex $\simplex_f$ and there are $2N$ faces dual to this face.

\begin{figure}
\begin{subfigure}{0.5\textwidth}
 \centering
\includegraphics[width=.9\linewidth]{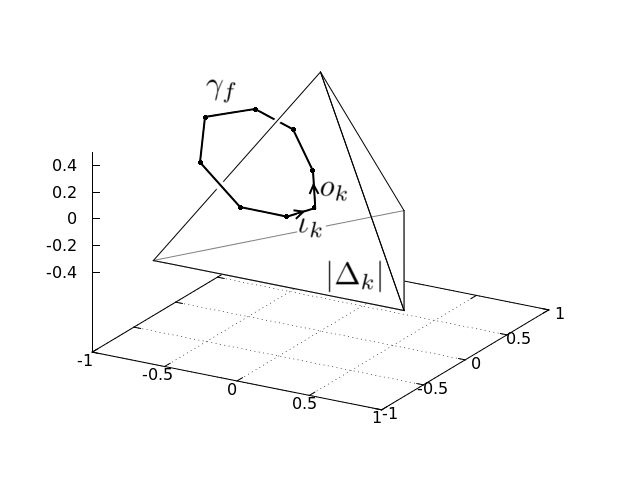}
  \caption{Curve $\gamma_{f}$. It is a composition of curves $\ine_k$ and $\oute_k$. Three-dimensional example ($n=3$).}
  \label{fig:path}
\end{subfigure}
\begin{subfigure}{0.5\textwidth}
 \centering
\includegraphics[width=.9\linewidth]{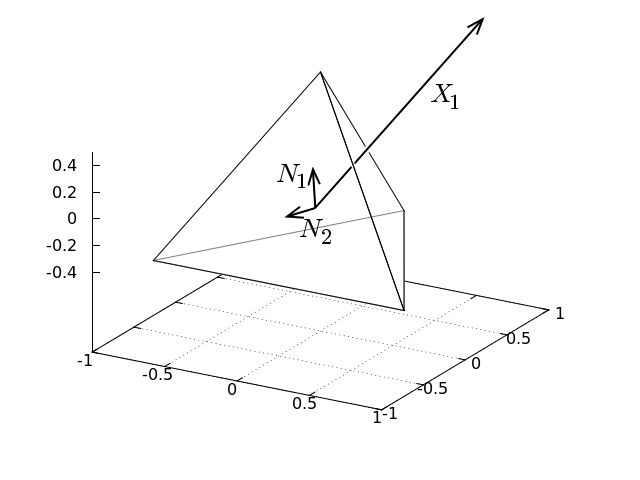}
  \caption{Vector fields $X_1,\ldots,X_{n-2},N_1,N_2$. Three-dimensional example ($n=3$).}
  \label{fig:vector_fields}
\end{subfigure}
\caption{We construct the action functional as a limit of discretized actions. We focus on 4-dimensional case but for simplicity we draw three dimensional examples. In 4 dimensions the curvature is discretized by infinitesimal holonomies $U_{\gamma_f}=\pm \exp(\curvature_f)$ taking values in ${\rm SO}(1,3)$. Each discretized action is a sum of face contributions $S_f[B,\connection]= s \Tr{B(X_1,X_{2}) \curvaturematrix_f}$, where $s=\sgn(\vol(X_1,X_{2},N_1,N_2))$.}
\label{fig:discretization_of_the_action}
\end{figure}
Being given a dual face $f$ we construct a curve $\gamma_{f}$ in $\manifold$ encircling the defect $\simplex_f$ (see figure \ref{fig:path}). For each $n$-simplex $\simplex_k$ in the sequence $f$ we construct two edges: one incoming to the barycenter of the $n$-simplex
$$
\ine_k(t)=\phi_k^{-1}\left( t\, b(|\simplex_k|) + (1-t)\, b(|\simplex_k\cap \simplex_{k-1}|_k)\right), \quad 0\leq t \leq 1,
$$
and one outgoing from the barycenter of the $n$-simplex:
$$
\oute_k(t)=\phi_k^{-1}\left( t\, b(|\simplex_k\cap \simplex_{k+1}|_k) + (1-t)\, b(|\simplex_k|)\right), \quad 0\leq t \leq 1.
$$
Let us define curves $s_{k+1,k}:[0,1]\to \manifold$, $k\in \mathbb{Z}_N$ as a compositions of $\oute_k$ with $\ine_{k+1}$:
$$
s_{k+1,k}=\ine_{k+1} \circ \oute_k.
$$
The loop $\gamma_{f}$ is a composition of such curves:
$$
\gamma_{f}=s_{1,N}\circ s_{N,N-1} \circ \ldots \circ s_{3,2}\circ s_{2,1}.
$$
Clearly $\gamma_{f}$ and $\gamma_{f^{-1}}$ are inverses of each other. 

We introduce the following vectors tangent to $\manifold$ at the barycenter of $\simplex_1$, i.e. at a point $p=\phi_1^{-1}(b(\gsimplex[1]))$:
$$
X_l(h):=\frac{d}{dt}h(\phi_1^{-1}(t(\sigma_1(v_l)-\sigma_1(v_0))+b(\gsimplex[1]) ))|_{t=0},
$$
where $v_0, v_l, l\in\{1,2,\ldots,n-2\}$ are the vertices of $\simplex_f$ and $h$ is a smooth function on some open neighbourhood of $p$. We also consider two vectors $N_1, N_2 \in T_p \manifold:$
\begin{align*}
N_1:=\dot{\oute}_1(0),\quad N_2(h):=-\dot{\ine}_1(1).
\end{align*}

From now on we focus on the physical case, when the pseudomanifold $\complex$ is \underline{4-dimensional} and the metric $\intmetric$ is the Minkowski metric
$$\intmetric=\diag(-1,1,1,1).$$ Let $\bundle$ be an SO($1,3$) principal fibre bundle over $\manifold$. We define an element of the so($1,3$) lie algebra $\curvaturematrix_{f}$ such that
$$
\exp(\curvaturematrix_{f}):=\begin{cases}U_{\gamma_f},{\rm\,if\ }U_{\gamma_f}\in{\rm SO}(1,3)^+,\\
-U_{\gamma_f},{\rm\, otherwise}.
\end{cases}
$$
where $U_{\gamma_f}$ is the SO($1,3$) matrix corresponding to the holonomy around the loop $\gamma_f$ in the trivialization over $U_1$. It will serve as a discretization of the curvature. Clearly, $\curvaturematrix_{f}$ depends on the choice of orientation of $\gamma_f$: $\curvaturematrix_{f^{-1}}=-\curvaturematrix_{f}$. We define the contribution to the action functional from the dual face in the following way:
$$
S_{f}[\B,\connection]= s \Tr{ B_1(X_1, X_{2}) \curvaturematrix_{f}},
$$
where $s=\sgn(\vol(X_1,X_{2},N_1,N_2))$. Clearly, $S_f$ is SO($1,3$) gauge invariant.  As we noted before for each $n$-simplex $\simplex_i\in {\complex}^{(n)}$ and its $(n-2)$-dimensional face $\simplex_f\isfaceof \simplex_i$ there correspond two sequences $f$ and $f^{-1}$. Let us note that $S_f$ does not depend on the orientation of the dual face $f$, i.e.
$$
S_{f}[B,\connection]=S_{f^{-1}}[B,\connection].
$$ 
If $f$ is a sequence of length $N$, then to the $(n-2)$-simplex $\simplex_f$ there correspond $2N$ different dual faces $f_1,\ldots,f_{2N}$. We define
\begin{equation*}
S_{\simplex_f}[B,\connection]:=\frac{1}{2N}\sum_{k=1}^{2N} S_{f_k}[B,\connection].
\end{equation*}
The action functional corresponding to $\complex$ is a sum of the functionals $S_{\simplex_f}$ (see also \cite{ThiemannBook,WielandDiscrete}):
$$
S_{\complex}[B,\connection]=\sum_{\simplex_f\in {\complex}^{(n-2)}} S_{\simplex_f}[B,\connection].
$$
We consider a regular refinement $\complex'_m$ of $\complex$. Since the manifold $\manifold$ corresponding to $\complex$ and $\manifold'$ corresponding to $\complex'_m$ can be identified $\bundle$ is also a principal fibre bundle over $\manifold'$. Our functional is a continuum limit of these discrete actions:
\begin{equation}\label{eq:action}
S[B,\connection]=\lim_{m\to \infty} S_{\complex'_m}[B,\connection].
\end{equation}
Since each $S_{f}$ is SO($1,3$) gauge invariant, so is $S$.

\section{Euler-Lagrange field equations}
\subsection{Boundary conditions}\label{sc:boundarydata}
We assume the following boundary conditions:
\begin{enumerate}
\item Each form $\phi_i^* \B_i,\, \phi_i^* \connection_i$ defined on $\phi_i(U_i)\subset \mathbb{R}^4$ extends smoothly to some neighborhood containing $\overline{\phi_i(U_i)}$; let us denote these extensions by $\tilde{\B}_i, \tilde{\connection}_i$.
\item We assume that the principal fibre bundle $\bundle$ is flat (i.e. supports a flat connection). We choose a trivialization in which the transfer functions are constant and we assume that for each 4-simplex $\simplex_i$ there exists a map $\mu_i:\simplex_i\to \mathbb{R}^n$ assigning a set of $n+1$ independent points in $\mathbb{R}^n$ to the set of vertices of a simplex $\simplex_i$, such that:
\begin{enumerate}
%\begin{itemize}
\item Each vector $\mu_i(v)-\mu_i(w)$, where $v,w \in \simplex_i$, is space-like, $v,w\in \simplex_i$. 
\item The orientation of the geometric simplex $\gsimplexm{\simplex_i}{i}:=\conv(\mu_i(\simplex_i))$ agrees with the orientation of $\mathbb{R}^n$:
$$
\vol[\complex]([v_0,v_1,v_2,v_3,v_4])={\rm sgn}(\epsilon_{I_1 I_2 I_3 I_4} (\mu_i(v_1)-\mu_i(v_0))^{I_1} \ldots (\mu_i(v_4)-\mu_i(v_0))^{I_4}).
$$
\item Being given a dual face $f=(\simplex_1,\ldots,\simplex_N)$, we fix an orientation of each $(n-2)$-simplex $\gsimplexm{\simplex_f}{i}:=\conv(\mu_i(\simplex_f))$: consider two constant vector fields $N_{1;i}$ and $N_{2;i}$ on $\mathbb{R}^4$ whose components are
$$N_{1;i}^I(p)=b(\gsimplexm{\simplex_i\cap \simplex_{i+1}}{i})^I-b(\gsimplexm{\simplex_i}{i})^I,$$ $$N_{2;i}^I(p)=b(\gsimplexm{\simplex_i\cap \simplex_{i-1}}{i})^I-b(\gsimplexm{\simplex_{i}}{i})^I.$$ 
The orientation $\gsimplexm{\simplex_f}{i}$ is defined by a volume form $\vol[\gsimplexm{\simplex_f}{i}]=\iota_{\gsimplexm{\simplex_f}{i}}^*(N_2 \lrcorner N_1\lrcorner \vol[4]),$
where $\vol[4]=dx^0\wedge \ldots \wedge dx^{3}$; 
$$\B_{f;i}^{IJ}=\int_{\gsimplexm{\simplex_f}{i}} \tilde{B}_i^{IJ}=\frac{1}{2}\tensor{\epsilon}{^{IJ}_{KL}} (\mu_{i}(v_1)-\mu_i(v_0))^K (\mu_{i}(v_2)-\mu_i(v_0))^L,$$
where $v_0,v_1,v_2$ are the vertices of $\simplex_f$ such that
$$
\epsilon_{IJKL} (\mu_{i}(v_1)-\mu_i(v_0))^I (\mu_{i}(v_2)-\mu_i(v_0))^J N_{1;i}^K(p) N_{2;i}^L(p) >0
$$
for arbitrary $p\in \mathbb{R}^n$.
%\end{itemize}
%\item $\tilde{\connection}_i(p)=0$ for each $p\in |\simplex_f|_i$, $\simplex_f\in \complex^{(2)}$, $\simplex_f\subset \simplex_i$.
\end{enumerate} 
\end{enumerate}
\begin{rmk}\label{rmk:faceorientation}
Let us note that the simplices $|\simplex_f|_i$ for $i$ ranging in the set of 4-simplices sharing $\simplex_f$ have consistent orientation, in the sense that if 
$$
\B_{f;i}^{IJ}=\frac{1}{2}\tensor{\epsilon}{^{IJ}_{KL}} (\mu_{i}(v_1)-\mu_i(v_0))^K (\mu_{i}(v_2)-\mu_i(v_0))^L
$$
then
$$
\B_{f;j}^{IJ}=\frac{1}{2}\tensor{\epsilon}{^{IJ}_{KL}} (\mu_{j}(v_1)-\mu_j(v_0))^K (\mu_{j}(v_2)-\mu_j(v_0))^L.
$$
In order to show it, let us note that the condition 
$$
\epsilon_{IJKL} (\mu_{i}(v_1)-\mu_i(v_0))^I (\mu_{i}(v_2)-\mu_i(v_0))^J N_{1;i}^K(p) N_{2;i}^L(p) >0
$$
is equivalent to 
$$
\epsilon_{IJKL} (\mu_{i}(v_1)-\mu_i(v_0))^I (\mu_{i}(v_2)-\mu_i(v_0))^J (\mu_{i}(v)-\mu_i(v_0))^K (\mu_{i}(v_3)-\mu_i(v_0))^L >0,
$$
where $\simplex_{i-1}\cap \simplex_{i}=\{ v_0,v_1,v_2,v\}$,  $\simplex_{i}\cap \simplex_{i+1}=\{ v_0,v_1,v_2,v_3\}$. Therefore $\vol[\complex]([v_0,v_1,v_2,v,v_3])>0$. From orientability condition of the simplicial complex, it follows, that $\vol[\complex]([v_0,v_1,v_2,v_3,w])>0$, where $\simplex_{i+1}=\{v_0,v_1,v_2,v_3,w\}$. Therefore
$$
\epsilon_{IJKL} (\mu_{i+1}(v_1)-\mu_{i+1}(v_0))^I (\mu_{i+1}(v_2)-\mu_{i+1}(v_0))^J (\mu_{i+1}(v_3)-\mu_i(v_0))^K (\mu_{i+1}(w)-\mu_{i+1}(v_0))^L >0.
$$
This equation is equivalent to
$$
\epsilon_{IJKL} (\mu_{i+1}(v_1)-\mu_{i+1}(v_0))^I (\mu_{{i+1}}(v_2)-\mu_{i+1}(v_0))^J N_{1;i+1}^K(p) N_{2;i+1}^L(p) >0.
$$
\end{rmk}

%\subsubsection{Geometric structure defined by the boundary conditions}
Although the constraints on the $B_{f;i}^{IJ}$ are imposed on each simplex separately, the geometries of the simplices match thanks to the implicit conditions
$$
\B_{f;i}^{IJ}=\tensor{g}{_{ij}^I_K} \tensor{g}{_{ij}^J_L} \B_{f;j}^{KL}.
$$
This is proved in the following theorem.
\begin{thm}\label{thm:geometricstructure}
The family of maps $\{\mu_i\}_{i\in \{1,\ldots,N}$ forms a geometric structure on $\complex$.
\end{thm}
\begin{proof}
It is enough to show that the maps $\mu_i$ are pairwise compatible. Let $\translation_v$ be a translation in $\mathbb{R}^4$ by a vector $v$:
$$
\translation_v(w)=w+v.
$$
Let us focus on two neighbouring 4-simplices $\simplex_i$ and $\simplex_j$. We claim that 
$$
\alpha_{ij}=\translation_{b(\gsimplexm{\simplex_i\cap\simplex_j}{i})}\, g_{ij}\, \translation_{-b(\gsimplexm{\simplex_i\cap\simplex_j}{j})}
$$
is the isometry of the Minkowski space $(\mathbb{R}^4,\intmetric)$ satisfying \eqref{eq:gluingpoints}. We will now prove this claim. Let $N_{ij}, N_{ji}\in \mathbb{R}^4$ be normalized vectors orthogonal to the boundary triangles of the tetrahedron $\gsimplexm{\simplex_i\cap \simplex_j}{j}$ and $\gsimplexm{\simplex_i\cap \simplex_j}{i}$, respectively. The vectors are unique up to transformations $N_{ij}\mapsto -N_{ij}$, $N_{ji}\mapsto -N_{ji}$. We fix this ambiguity by requiring that 
\begin{equation}\label{eq:orientationi}\epsilon_{IJKL} (\mu_j(v_1)-\mu_j(v_0))^I (\mu_j(v_2)-\mu_j(v_0))^J (\mu_j(v_3)-\mu_j(v_0))^K N_{ij}^L <0\end{equation}
and
\begin{equation}\label{eq:orientationii}\epsilon_{IJKL} (\mu_i(v_1)-\mu_i(v_0))^I (\mu_i(v_2)-\mu_i(v_0))^J (\mu_i(v_3)-\mu_i(v_0))^K N_{ji}^L >0,\end{equation}
where $\epsilon_{IJKL}$ is the Levi-Civita symbol, $\epsilon_{0123}=1$, $\simplex_i=\{v_0,v_1,v_2,v_3,v\}, \simplex_j=\{v_0,v_1,v_2,v_3,w\}$ and $\vol[\complex]([v_0,v_1,v_2,v_3,v])>0,\, \vol[\complex]([v_0,v_1,v_2,v_3,w])<0 $. In other words $N_{ij}$ is an outward pointing normal to the tetrahedron $\gsimplexm{\simplex_i \cap \simplex_j}{j}$ and $N_{ji}$ is an outward pointing normal to the tetrahedron $\gsimplexm{\simplex_i\cap \simplex_j}{i}$. Consider two SO(1,3) transformations $h_i,\,h_j$, mapping the vectors $N_{ij}$ and $N_{ji}$ into $(1,0,0,0)$ and $(-1,0,0,0)$, respectively. Define maps $\mu'_i:\simplex_i\to \mathbb{R}^4$, $\mu'_j:\simplex_j\to \mathbb{R}^4$ by the following formula
$$
\mu'_i(v) = h_i (\mu_i(v) - b(\gsimplexm{\simplex_i\cap\simplex_j}{i})),\quad \mu'_j(v) = h_j (\mu_j(v) - b(\gsimplexm{\simplex_i\cap\simplex_j}{j})).
$$
The maps $\mu'_i$ and $\mu'_j$ are obtained from $\mu_i$ and $\mu_j$ by composing with affine isometries. Therefore the relation can be easily inverted. The conditions \eqref{eq:orientationi} and \eqref{eq:orientationii} become:
\begin{equation}\label{eq:orientationij}
\epsilon_{abc} (\mu'_j(v_1)-\mu'_j(v_0))^a (\mu'_j(v_2)-\mu'_j(v_0))^b (\mu'_j(v_3)-\mu'_j(v_0))^c >0
\end{equation}
and
\begin{equation}\label{eq:orientationji}
\epsilon_{abc} (\mu'_i(v_1)-\mu'_i(v_0))^a (\mu'_i(v_2)-\mu'_i(v_0))^b (\mu'_i(v_3)-\mu'_i(v_0))^c >0,
\end{equation}
where $a,b,c\in\{1,2,3\}$, $\epsilon_{123}=1$.
%The maps $\mu'_i,\,\mu'_j$ have the following properties
%$$(\mu'_i(v)-\mu'_i(w))^0=0,\quad (\mu'_j(v)-\mu'_j(w))^0=0$$
%for $v,w\in \simplex_i\cap\simplex_j$. Moreover
Let us note that
$$
\mu_i(v)=\alpha_{ij}(\mu_j(v)) \iff  \mu_i'(v)=\alpha'_{ij}(\mu_j'(v)),
$$
where $\alpha'_{ij}=h_i g_{ij} h_j^{-1}$. Therefore, in order to check the property \eqref{eq:gluingpoints} it is enough to check the primed version from the equation above. Let us note that
$$
\B_{f;i}^{IJ}=\tensor{g}{_{ij}^I_K} \tensor{g}{_{ij}^J_L} \B_{f;j}^{KL}.
$$
Since the normal vectors $N_{ij}$ and $N_{ji}$ are defined uniquely up to overall $\pm1$ factor by equations:
\begin{equation}\label{eq:orthogonality}
N_{ji}^J \epsilon_{IJKL} B_{f;i}^{KL}=0,\quad N_{ij}^J \epsilon_{IJKL} B_{f;j}^{KL}=0
\end{equation}
it follows that $N_{ji}=\pm g_{ij} N_{ij}$. As a result, $\alpha'_{ij}$ is a linear isometry of $\mathbb{R}^4$ which maps vectors orthogonal to $(1,0,0,0)$ into vectors orthogonal to $(1,0,0,0)$. Therefore, when restricted to $\mathbb{R}^3$, it is an orthogonal transformation. Denote by
$$
{\B'}_{f;i}^{IJ} := \tensor{(h_i)}{^I_K} \tensor{(h_i)}{^J_L} {\B}_{f;i}^{KL}, \quad {\B'}_{f;j}^{IJ} := \tensor{(h_j)}{^I_K} \tensor{(h_j)}{^J_L} {\B}_{f;j}^{KL}.
$$
In particular (see remark \ref{rmk:faceorientation}):
\begin{eqnarray}
{\B'}_{f;i}^{IJ} = \frac{1}{2}\tensor{\epsilon}{^{IJ}_{KL}} (\mu_{i}'(w_1)-\mu_i'(w_0))^K (\mu_{i}'(w_2)-\mu_i'(w_0))^L,\label{eq:simplexfpi} \\ {\B'}_{f;j}^{IJ} = \frac{1}{2}\tensor{\epsilon}{^{IJ}_{KL}} (\mu_{j}'(w_1)-\mu_j'(w_0))^K (\mu_{j}'(w_2)-\mu_j'(w_0))^L,\label{eq:simplexfpj}
\end{eqnarray}
where $\simplex_f=\{w_0,w_1,w_2\}$. Denote also by
$$
\gsimplexmp{\simplex_i\cap\simplex_j}{k}= \conv (\mu'_k(\simplex_i\cap\simplex_j)),\quad \gsimplexmp{\simplex_f}{k}:=\conv (\mu'_k(\simplex_f)), k\in\{i,j\}.
$$
From \eqref{eq:orthogonality} it follows that 
$$
{\B'}_{f;i}^{ab}={\B'}_{f;j}^{ab}=0,\ a,b\in \{1,2,3\}.
$$
Consider vectors
$$
K_{f;k}^a:=\begin{cases}{\B'}_{f;k}^{0a},\textrm{\,if the orientation of }\gsimplexmp{\simplex_f}{k}\textrm{ agrees with the orientation of }\gsimplexmp{\simplex_i\cap\simplex_j}{k},\\ -{\B'}_{f;k}^{0a},\textrm{\,if the orientation of }\gsimplexmp{\simplex_f}{k}\textrm{ is opposite to the orientation of }\gsimplexmp{\simplex_i\cap\simplex_j}{k},\end{cases}
$$
where $k=i$ or $k=j$. The vectors $K_{f;i}^a$ and $K_{f;j}^a$ are outward pointing normals to the boundary triangles of the tetrahedra $\gsimplexmp{\simplex_i\cap\simplex_j}{i}$ and $\gsimplexmp{\simplex_i\cap\simplex_j}{j}$ respectively. From equations \eqref{eq:orientationij}, \eqref{eq:orientationji}, \eqref{eq:simplexfpi} and \eqref{eq:simplexfpj} it follows that the orientation of $\gsimplexmp{\simplex_f}{i}$ agrees with the orientation of $\gsimplexmp{\simplex_i\cap\simplex_j}{i}$ if and only if the orientation of $\gsimplexmp{\simplex_f}{j}$ agrees with the orientation of $\gsimplexmp{\simplex_i\cap\simplex_j}{j}$. As a result the map $\alpha_{ij}'$ maps the outward pointing normals to the faces of the tetrahedron $\gsimplexmp{\simplex_i\cap\simplex_j}{i}$ into outward pointing normals to the faces of the tetrahedron $\gsimplexmp{\simplex_i\cap\simplex_j}{j}$. Being an orthogonal transformation, it preserves the lengths of vectors:
$$
|| \vec{K}_{f;i}|| = || \vec{K}_{f;j}||,
$$
Therefore the area of a triangle $\gsimplexmp{\simplex_f}{i}$ is the same as the area of the triangle $\gsimplexmp{\simplex_f}{j}$. From Minkowski theorem about convex polyhedra \cite{Minkowski} it follows that $\alpha_{ij}'(\gsimplexmp{\simplex_i\cap\simplex_j}{j})$ differs from $\gsimplexmp{\simplex_i\cap\simplex_j}{i}$  possibly by a translation. However, $\alpha_{ij}'$ maps a barycenter of $\gsimplexmp{\simplex_i\cap\simplex_j}{j}$ into barycenter of $\gsimplexmp{\simplex_i\cap\simplex_j}{i}$. Therefore 
\begin{equation}\label{eq:equalityoftetrahedra}
\gsimplexmp{\simplex_i\cap\simplex_j}{i} = \alpha_{ij}'(\gsimplexmp{\simplex_i\cap\simplex_j}{j}).
\end{equation}
Since it maps a normal to a face $\gsimplexmp{\simplex_f}{j}$, to a normal of a face $\gsimplexmp{\simplex_f}{i}$, it maps the points of $\gsimplexmp{\simplex_f}{j}$ to the points of $\gsimplexmp{\simplex_f}{i}$, i.e.
\begin{equation}\label{eq:equalityoffaces}
 \gsimplexmp{\simplex_f}{i}=\alpha'_{ij}(\gsimplexmp{\simplex_f}{j}).
\end{equation}
For each $v\in\simplex_i\cap\simplex_j$ there is unique triangle $\simplex_f\subset \simplex_i\cap \simplex_j$ such that $v\not\in \simplex_f$. From equations \eqref{eq:equalityoftetrahedra} and \eqref{eq:equalityoffaces} it follows that 
$$
\mu_i'(v)=\alpha'_{ij} (\mu_j'(v)).
$$ 
\end{proof}
\subsection{The Euler-Lagrange field equations for the BF theory on manifolds with defects}\label{sc:ELBFdefects}
We will now derive the Euler-Lagrange field equations characterizing the stationary points of the action functional \eqref{eq:action}. We will consider variations of the $B$ field preserving these boundary conditions, i.e. variations $\delta \connection, \delta B$ such that $\delta B$ vanishes at the boundary. We do not assume vanishing of $\delta \connection$ at the boundary. We require only that each local Lie-algebra-valued 1-form $\phi^*_i\delta \connection_i$ extends smoothly to some neighborhood containing $\overline{\phi_i(U_i)}$.

We split the set of $(n-2)$-simplices of a subdivision $\complex'_m$ of $\complex$ into two disjoint sets $${\complex'}^{(n-2)}_m={\complex'}^{\rm interior}_m\cup {\complex'}^{\rm boundary}_m,$$
where ${\complex'}^{\rm boundary}_m$ consists of those $(n-2)$-simplices of $\complex'_m$ that are contained in $(n-2)$-simplices of $\complex$, i.e. for each $(n-2)$-simplex $\simplex'_{f'}\isfaceof \simplex'_{i'}$ there exists an $(n-2)$-simplex $\simplex_f\isfaceof \simplex_i$ such that
$$
\varphi(|\simplex'_{f'}|_{i'})\subset |\simplex_f|_i.
$$

The action functional $S_{\complex'_m}[\B,\connection]$ can be split into a sum of two contributions 
$$S_{\complex'_m}[\B,\connection]=S^{\rm interior}_{\complex'_m}[\B,\connection]+S^{\rm boundary}_{\complex'_m}[\B,\connection],$$
where 
$$
S^{\rm boundary}_{\complex'_m}[\B,\connection]:=\sum_{\simplex'_{f'}\in {\complex'}^{\rm boundary}_m} S_{\simplex'_{f'}}[\B,\connection].
$$
This splitting of the functional $S_{\complex'_m}$ leads to a splitting of the action functional \eqref{eq:action}:
$$
S[\B,\connection]=S^{\rm interior}[\B,\connection]+S^{\rm boundary}[\B,\connection],
$$
where 
$$
S^{\rm interior}[\B,\connection]=\lim_{m\to \infty} S^{\rm interior}_{\complex'_m}[\B,\connection]
$$
and
$$
S^{\rm boundary}[\B,\connection]=\lim_{m\to \infty} S^{\rm boundary}_{\complex'_m}[\B,\connection].
$$
We will assume that the fields $\B$ and $\connection$ are such that the form $\Tr{\B\wedge \curvature}$ is Lebesgue integrable. With this assumption the part  $S^{\rm interior}[\B,\connection]$ can be expressed as an integral\footnote{Let us notice that Riemann integrability would be not sufficient.}:
$$
S^{\rm interior}[\B,\connection]=\int_{\manifold}\Tr{\B\wedge \curvature}.
$$
We calculate the Euler-Lagrange field equations. First, let us notice that
\begin{thm}
$$
S^{\rm boundary}[\B+\delta B,\connection+\delta \connection] = S^{\rm boundary}[\B,\connection]
$$
for any variations $\delta B$ and $\delta \connection$ such that $\delta B=0$ at the boundary.
\end{thm}
\begin{proof} To show this notice that
\begin{align*}
S^{\rm boundary}_{\complex'_m}[\B+\delta B,\connection+\delta \connection]-S^{\rm boundary}_{\complex'_m}[\B,\connection]=\\=\sum_{\simplex'_{f'}\in {\complex'}^{\rm boundary}_m}\left( \pm\delta B (X_1,X_{2}) \curvaturematrix_f(\connection+\delta \connection) \pm  B (X_1,X_{2}) (\curvaturematrix_f(\connection+\delta \connection) - \curvaturematrix_f(\connection)) \right).
\end{align*}
Let $\epsilon=\mesh(\complex'_m)$. We notice that $X_1,X_2=\bigo(\epsilon)$. Since $\delta B$ vanishes at the boundary, it follows that $\delta B(X_1,X_2)=\bigo(\epsilon^3)$. On the other hand $\curvaturematrix_f=\bigo(1)$. From the assumptions about the subdivision stated in Section \ref{sc:action_functional} it follows that $\# {\complex'}^{\rm boundary}_m = \bigo(\epsilon^{-2})$. Therefore 
$$
\sum_{\simplex'_{f'}\in {\complex'}^{\rm boundary}_m}\delta B (X_1,X_{2}) \curvaturematrix_f=\bigo(\epsilon).
$$
Now we will show that $\curvaturematrix_f(\connection+\delta \connection) - \curvaturematrix_f(\connection)=\bigo(\epsilon)$. Let us for simplicity assume that all the transition functions around an $(n-2)$-simplex $\simplex'_{f'}\in {\complex'}^{\rm boundary}_m$ are identity matrices except for $g_{1 N}$ \footnote{It is always possible to use an equivalent fibre bundle such that this holds.}. In this case the holonomy around the loop $\gamma_f$ is simply
$$
U_{\gamma_f}(\connection)=g_{1 N} \mathcal{P}\exp(\int_{\gamma_f} \connection ).
$$
Let us denote by $\curvaturematrix_{1N}$ the Lie algebra element such that $\exp(\curvaturematrix_{1N})=\pm g_{1 N}$. Let us recall that the curve $\gamma_f$ is composed from segments $\ine_{k},\oute_k$. Let us notice that $\dot{\ine}_k(1),\dot{\oute}_k(0)=\bigo(\epsilon)$. From the parallel transport equation it immediately follows that:
\begin{align*}
\mathcal{P}\exp(\int_{\oute_k} \connection )=\mathbbm{1}+ \connection(\dot{\oute}_k(0)) + \bigo(\epsilon^2)=\exp(\connection(\dot{\oute}_k(0)))+\bigo(\epsilon^2)=\\=\exp(\connection(\dot{\oute}_k(0)))(\mathbbm{1}+\bigo(\epsilon^2))=\exp(\connection(\dot{\oute}_k(0))+\ln(\mathbbm{1}+\bigo(\epsilon^2)))=\\=\exp(\connection(\dot{\oute}_k(0))+\bigo(\epsilon^2))=\exp(\int_{\oute_k} \connection + \bigo(\epsilon^2)),
\end{align*}
where in the last equality we used the fact that the left Riemann sum (consisting of only one summand) $\connection(\dot{\oute}_k(0))$ approximates the integral $\int_{\oute_k} \connection$ up to an error of order $\bigo(\epsilon^2)$. Using the Baker-Campbell-Hausdorff formula and the fact that the number of segments $\ine_k,\oute_k$ in a loop $\gamma_f$ is bounded from above, we see that
$$
\mathcal{P}\exp(\int_{\gamma_f} \connection )=\exp(\int_{\gamma_f} \connection + \bigo(\epsilon^2)).
$$
Applying the Baker-Campbell-Hausdorff formula again we obtain
$$
U_{\gamma_f}(\connection)=\pm\exp( \curvaturematrix_{1N} + {\mathcal A} \int_{\gamma_f} \connection + \bigo(\epsilon^2)),
$$
where ${\mathcal A}$ is a linear operator ${\mathcal A}:\mathfrak{g}\to \mathfrak{g}$ not depending on $\epsilon$ defined by iterated adjoint actions of the Lie algebra element $\curvaturematrix_{1N}$. 
Since the (Euclidean) length of the loop $\gamma_f$ is of order $\epsilon$, it follows that $\int_{\gamma_f} \delta \connection = \bigo(\epsilon)$. Thus we conclude that
$$
\curvaturematrix_f(\connection+\delta \connection) - \curvaturematrix_f(\connection)= {\mathcal A}\int_{\gamma_f} \delta \connection + \bigo(\epsilon^2)= \bigo(\epsilon).
$$
Let us notice, that we did not need to assume that $\delta \connection=0$ at the boundary. Since $B(X_1,X_2)=\bigo(\epsilon^2)$ and $\# {\complex'}^{\rm boundary}_m = \bigo(\epsilon^{-2})$, we see that
$$
 \sum_{\simplex'_{f'}\in {\complex'}^{(n-2)}_{\rm boundary}} \pm B (X_1,X_{2}) (\curvaturematrix_f(\connection+\delta \connection) - \curvaturematrix_f(\connection)) = \bigo(\epsilon).
$$
Therefore in the limit $\epsilon\to 0$ we obtain that
$$
S^{\rm boundary}[\B+\delta B,\connection+\delta \connection] - S^{\rm boundary}[\B,\connection]=0.
$$
\end{proof}
An immediate consequence of this theorem is that the stationary points of the action \eqref{eq:action} satisfying the boundary conditions coincide with the stationary points of $S^{\rm interior}[\B,\connection]=\int_{\manifold} \Tr{\B\wedge \curvature}$. In particular, the Euler-Lagrange field equations are:
$$
\D \B=0,\quad \curvature=0
$$
on $\manifold$. Since the solutions of these field equations are flat connections, the part $S^{\rm interior}$ vanishes on the solutions $\mathring{\B},\mathring{\connection}$ of these field equations. We will choose a trivialization in which the transition functions are constant. The remaining contribution comes from the boundary:
$$
S[\mathring{\B},\mathring{\connection}]=S^{\rm boundary}[\mathring{\B},\mathring{\connection}]=\sum_{\simplex_f \in \complex^{(2)}} \Tr{ B_{f;i} \curvaturematrix_{f;i}},
$$
where $\simplex_f$ is the unique triangle corresponding to a face $f$ and $\simplex_i$ is any simplex in the sequence $f$ (due to the transformation properties of $\B_f$ and $U_f$ the expression $\Tr{ B_{f;i} \curvaturematrix_{f;i}}$ does not depend on the choice of simplex $\simplex_i$ in the sequence $f$).
\subsection{The canonical $B$-field and flat connection on a manifold with defects}
Let $\complex$ be 4-dimensional oriented pseudomanifold and let $\mu$ be a Lorentzian geometric structure such that each vector $\mu_i(v)-\mu_i(w), v,w\in \simplex_i$ is space-like. In this section we will show that there exists a flat principal fibre bundle $\bundle$ and solutions of the Euler-Lagrange field equations satisfying boundary conditions defined by $\mu$ (see Section \ref{sc:boundarydata}, in particular Theorem \ref{thm:geometricstructure}). 

We consider the corresponding manifold with defects $\manifold$ and a bundle $L(\manifold)$ of linear frames over $\manifold$. The structure group of $L(\manifold)$ is reducible to O(1,3), because the transition functions $\tmap$ are Jacobi matrices of affine isometries of the Minkowski space $(\mathbb{R}^4,\intmetric)$:
$$
\tmap=\D\alpha_{ij}.
$$  Since $\manifold$ is orientable, the structure group can be further reduced to SO(1,3). The resulting subbundle of oriented orthonormal frames will be our principal fibre bundle $\bundle$. We will denote by ${\rm SO}(1,3)\times_\rho \mathbb{R}^4$ the fibre bundle associated to $\bundle$ via the defining representation $\rho$. Let $\connection$ be a connection form on $\bundle$ and $\triad$ be a ${\rm SO}(1,3)\times_\rho \mathbb{R}^4$-valued 1-form, called a vierbein. It is well known that a connection form defines and is uniquely defined by a family $\{\connection_i\}$ of Lie algebra valued 1-forms each defined on $U_i$ satisfying the following conditions (see for example Chapter 2, Proposition 1.4. in \cite{KobayashiNomizuI}):
$$
\connection_i=\tmap\, \connection_j\, \tmap^{-1} - d\tmap\, \tmap^{-1} {\rm\ on\ }U_i\cap U_j.
$$
Similarly, a vierbein $\triad$ is described by a family $\{\triad_i\}$ of $\mathbb{R}^n$-valued 1-forms each defined on $U_i$ subject to the condition:
$$
\triad_i = \rho(g_{ij}) \triad_j  {\rm\ on\ }U_i\cap U_j.
$$

The canonical vierbein is given by:
\begin{align*}
\mathring{e}_i^I=dx^I, {\rm\ if\ } dx^1\wedge dx^2\wedge dx^3\wedge dx^4=f\, \vol|_{U_i},\ \forall_{p\in\manifold} f(p)>0,\\
\mathring{e}_i^1=dx^1,\mathring{e}_i^2=dx^2 ,\mathring{e}^{3}=dx^4,\,\mathring{e}^{4}=dx^{3}, {\rm\ if\ } dx^1\wedge dx^2\wedge dx^3\wedge dx^4=f\, \vol|_{U_i},\ \forall_{p\in\manifold} f(p)<0.
\end{align*}
As always $x^I$ are the component functions of a coordinate map $\phi_i$:
$$
\phi_i(p)=(x^1(p),\ldots,x^4(p)).
$$
The canonical $B$-field satisfying the boundary conditions defined by $\mu$ is
$$
\mathring{B}^{IJ}=\star (\mathring{\triad}^I\wedge \mathring{\triad}^J).
$$
There is a canonical connection on $\bundle$:
$$
\mathring{\connection}_i = 0.
$$
Clearly, the connection $\mathring{\connection}$ is flat
$$
\curvature=0
$$
and compatible with the vierbein $\mathring{\triad}$:
$$
d \mathring{e}_i^I+\tensor{\mathring{\omega}}{_i^I_J}\wedge \mathring{e}_i^J=d\mathring{e}_i^I =0.
$$
In particular,
$$
\D \mathring{B}=0.
$$
\begin{rmk}
Although the connection $\mathring{\connection}$ is flat, a holonomy around a defect can be non-trivial. Let us consider for example an $n-2$ simplex $\simplex_f$ and its dual face $f=(\simplex_1,\ldots,\simplex_N)$. The holonomy around the loop $\gamma_f$ is
$$
U_{\gamma_f}=g_{1\, 2}\ \ldots\ g_{N-1\,N}\ g_{N\,1},
$$
where $g_{ij}$ are the transition functions. Let us note that the transition functions need to be constant, because $\mathring{\connection}_i=0$. These equations are often interpreted in the literature as a discretization of the connection \cite{Bander,Khatsymovsky,Baezintro}. In our approach the connection is smooth (not discretized) and $g_{ij}$ are transition functions defining an (in general non-trivial) principal fibre bundle.
\end{rmk}
\section{Relation between Regge calculus and BF theory on manifolds with defects}
\subsection{Regge action}
Since we focus on the (physical) Lorentzian signature, we use a Lorentzian version of the Regge action introduced in \cite{BarrettFoxon} (see also \cite{BarrettAsymptotics}). We will review it in this subsection.

Let $\complex$ be a 4-dimensional oriented pseudomanifold. Regge data is an assignment to each 1-simplex $e\in \complex^{(1)}$, called an edge, its length $\length(e)$.  A Lorentzian geometric structure $\sigma$ on $\complex$ defines the length of each edge $e=\{v_0,v_1\}$ as the unique positive number $\length(e)$ such that
$$
\length(e)^2=\intmetric(\sigma_i(v_1)-\sigma_i(v_0),\sigma_i(v_1)-\sigma_i(v_0)),
$$
where $\simplex_i$ is any $n$-simplex which face is the edge $e$, $e\isfaceof \simplex_i$. Let $N_{ij}$ be the outward pointing normal to the tetrahedron $|\simplex_{i}\cap \simplex_{j}|_j$. We can further classify the outward pointing normals as future pointing or past pointing according to the standard time-orientation of the Minkowski space. Following \cite{BarrettFoxon,BarrettAsymptotics} we define a dihedral angle between tetrahedra $|\simplex_{i}\cap \simplex_{j}|_i$ and $|\simplex_{i}\cap \simplex_{k}|_i$ at the triangle $\simplex_f=\simplex_i\cap\simplex_j\cap \simplex_k$ (we assume that $\simplex_i\cap\simplex_j\cap \simplex_k$ is a 2-simplex) to be (see figure \ref{fig:dihedral_angle}):
\begin{itemize}
\item the positive angle $\dihedral_{f,i}$ such that
$$
\cosh(\dihedral_{f,i})=N_{ji}\cdot N_{ki},
$$
if one of the normals is future-pointing and the other is past-pointing;
\item the negative angle $\dihedral_{f,i}$ such that
$$
\cosh(\dihedral_{f,i})=-N_{ji}\cdot N_{ki},
$$
if both normals are either future-pointing or past-pointing.
\end{itemize}
\begin{figure}
\begin{subfigure}{0.5\textwidth}
 \centering
\includegraphics[width=.9\linewidth]{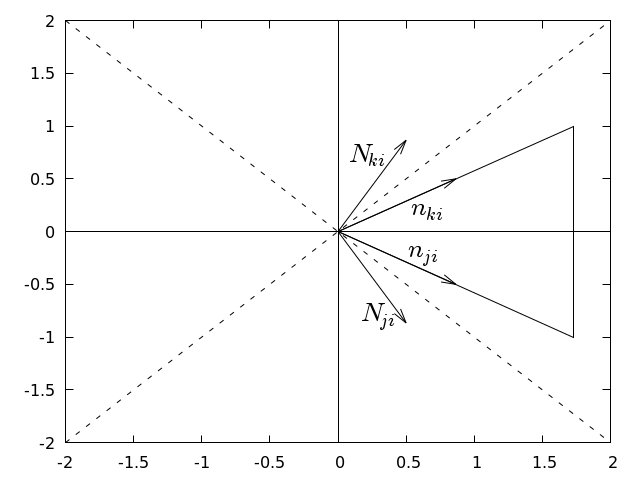}
  \caption{Thin wedge.}
  \label{fig:thin_wedge}
\end{subfigure}
\begin{subfigure}{0.5\textwidth}
 \centering
\includegraphics[width=.9\linewidth]{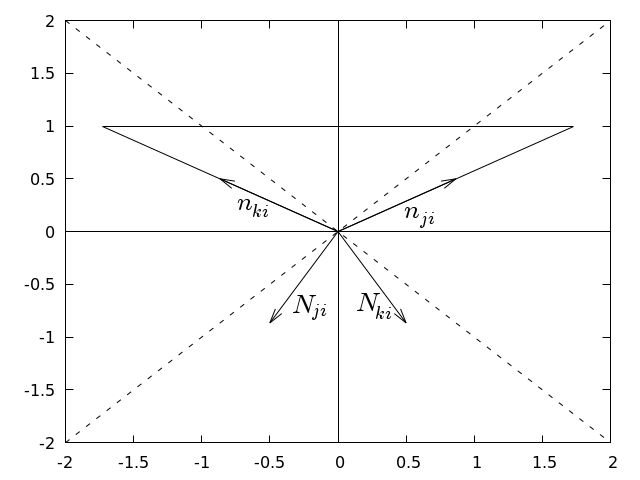}
  \caption{Thick wedge.}
  \label{fig:thick_wedge}
\end{subfigure}
\caption{We define a dihedral angle $\dihedral_{f,i}$ between tetrahedra $|\simplex_{i}\cap \simplex_{j}|_i$ and $|\simplex_{i}\cap \simplex_{k}|_i$ at the triangle $\simplex_f=\simplex_i\cap\simplex_j\cap \simplex_k$. The vector $N_{ij}$ is the unit outward pointing normal to the tetrahedron $|\simplex_{i}\cap \simplex_{j}|_j$. The vector $n_{ij}$ is the unit vector parallel to the tetrahedron $|\simplex_i\cap \simplex_j|_j$, orthogonal to the triangle $|\simplex_f|_j$ and pointing inside the tetrahedron. The dihedral angle $\dihedral_{f,i}$ is the positive angle such that $\cosh(\dihedral_{f,i})=N_{ji}\cdot N_{ki}$ in case (a) and the negative angle $\dihedral_{f,i}$ such that $\cosh(\dihedral_{f,i})=-N_{ji}\cdot N_{ki}$ in case (b).}
\label{fig:dihedral_angle}
\end{figure}
%\marginpar{Sprawdzic definicje $n_{ij}$ i $N_{ij}$.}
Following \cite{BarrettFoxon} we will say that tetrahedra $|\simplex_{i}\cap \simplex_{j}|_i$ and $|\simplex_{i}\cap \simplex_{k}|_i$ form a thin wedge at the triangle $|\simplex_f|_i$ if $\dihedral_{f,i}$ is positive or thick wedge if it is negative. 

We define the deficit angle to be the negative of the sum of dihedral angles at a triangle $\simplex_f$%\marginpar{Sprawdzic negative}
$$
\deficit_f= -\sum_{i: \simplex_i\cap \simplex_f \neq \emptyset} \dihedral_{f,i}.
$$
Let $\area_f$ be the area of any triangle $|\simplex_f|_i$ (let us note that the area is the same for any $i$). The Regge action takes the following form:
$$
S_{\rm Regge}(\complex,\length)= \sum_{\simplex_f \in \complex^{(2)}} \area_f \deficit_f.
$$

\subsection{Relation between Regge action and BF theory on manifolds with defects}
Remarkably, the evaluation of the action action functional \eqref{eq:action} on the solutions of its Euler-Lagrange field equations satisfying the boundary conditions from Section \ref{sc:boundarydata} coincides with the evaluation of the Regge action at Regge variables defined by the boundary data. This is shown by the following theorem.
\begin{thm}
Let $\mathring{\B},\mathring{\connection}$ be the solutions of the Euler-Lagrange field equations satisfying the boundary conditions defined by a family of maps $\mu_i:\simplex_i\to \mathbb{R}^4$ (see Section \ref{sc:boundarydata}). Denote by $\length:\complex^{(1)}\to \mathbb{R}_+$ the map assigning to each edge $e=\{v_0,v_1\}$ the unique positive number $\length(e)$ such that
$$
\length(e)^2=\intmetric(\mu_i(v_1)-\mu_i(v_0),\mu_i(v_1)-\mu_i(v_0)),\quad e\isfaceof \simplex_i.
$$
The following equality holds:
$$
\frac{1}{2}S[\mathring{\B},\mathring{\connection}]=S_{\rm Regge}(\complex,\length).
$$
\end{thm}
\begin{proof}
In order to show it, it is enough to check that
$$
\Tr{ B_{f;i} \curvaturematrix_{f;i}}=2\area_f \deficit_f.
$$
Without loss of generality, we can assume, that the $n$-simplices are numbered such that $f=(\simplex_1,\ldots,\simplex_N)$. Since $\Tr{ B_{f;i} \curvaturematrix_{f;i}}$ is gauge invariant, we will work in a convenient gauge in which
$$g_{1\,2}=g_{2\,3}=\ldots=g_{N-1\,N}=\identity,$$
where $\identity$ is the identity matrix. In fact, without loss of generality we can assume that 
$$
\mu_{i}(v)=\mu_j(v),
$$
whenever $v\in \simplex_i\cap \simplex_j,\, i,j\in\{1,\ldots,N\}$\footnote{To this end we use an equivalent geometric structure obtained by appropriate translations.}. Let $\simplex_f=\{v_0,v_1,v_2\}$ and assume that the orientation of $|\simplex_f|_i$ is such that:
$$
B_{f;i}^{IJ}=\frac{1}{2}\tensor{\epsilon}{^{IJ}_{KL}} (\mu_{i}(v_1)-\mu_i(v_0))^K (\mu_{i}(v_2)-\mu_i(v_0))^L.
$$
We will denote by the same symbol $B_{f;i}$ the Lie algebra element
$$
\tensor{B}{_{f;i}^I_J}=B_{f;i}^{IK}\intmetric_{KJ}.
$$
and by $\star B_{f;i}$ the Lie algebra element
$$
\tensor{\star B}{_{f;i}^I_J}=\frac{1}{2}\tensor{\epsilon}{^I_{JKL}}B_{f;i}^{KL}.
$$
We consider a matrix
$$
\exp(-\frac{\dihedral_{f,i}}{\area_f} B_{f;i}).
$$
Let $n_{ij}$ be the vector parallel to the tetrahedron $|\simplex_i\cap \simplex_j|_j$, orthogonal to the triangle $|\simplex_f|_j$ and pointing inside the tetrahedron, i.e.
$$
n_{ij} \cdot (\mu_j(v_3)-\mu_j(v_0))>0,\, \textrm{where }\simplex_i\cap\simplex_j=\{v_0,v_1,v_2,v_3\}.
$$
Let $N_{ij}$ be the outward pointing normal to the tetrahedron $|\simplex_i\cap \simplex_j|_j$ (considered as a boundary tetrahedron of $|\simplex_j|$). Denote by $w_i,\, i\in\{1,\ldots,N\},$ the unique vertex of $\simplex_i$ such that $w_i \notin \simplex_i\cap \simplex_{i+1}$. Using this notation we have
$$
\simplex_i=\{v_0,v_1,v_2,w_i,w_{i+1}\},
$$
where $w_{N+1}:=w_1$. Let us note that
$$
\vol[\complex]([v_0,v_1,v_2,w_i,w_{i+1}])=1.
$$
This means in particular that
$$
n_{i-1\,i}^I=- \frac{1}{\area_f}\tensor{B}{_{f;i}^I_J} N_{i-1\,i}^J,
$$
for $i\in\{1,2,\ldots,N\}$, where we used a notation $n_{0\,1}:=n_{N\,1},\,N_{0\,1}:=N_{N\,1}$.
A straightforward calculation shows that
$$
N_{i-1\,i}^I= - \frac{1}{\area_f} \tensor{B}{_{f;i}^I_J} n_{i-1\,i}^J.
$$
Using these definitions we can easily calculate $\exp(-\frac{\dihedral_{f,i}}{\area_f}B_{f;i})N_{i-1\,i}$ by expanding into power series:
$$
\exp(-\frac{\dihedral_{f,i}}{\area_f}B_{f;i})N_{i-1\,i}=\sum_{k=0}^{\infty} \frac{(\dihedral_{f,i})^{2k}}{(2k)!} N_{i-1\,i}+\sum_{k=0}^{\infty} \frac{(\dihedral_{f,i})^{2k+1}}{(2k+1)!} n_{i-1\,i}.
$$
As a result
\begin{equation}\label{eq:boostmatrixN}
\exp(-\frac{\dihedral_{f,i}}{\area_f}B_{f;i})N_{i-1\,i}=\cosh(\dihedral_{f,i}) N_{i-1\,i}+\sinh(\dihedral_{f,i}) n_{i-1\,i}.
\end{equation}
Similarly
\begin{equation}\label{eq:boostmatrixn}
\exp(-\frac{\dihedral_{f,i}}{\area_f}B_{f;i})n_{i-1\,i}=\sinh(\dihedral_{f,i}) N_{i-1\,i}+\cosh(\dihedral_{f,i}) N_{i-1\,i}.
\end{equation}

We will consider now separately the two possibilities:
\begin{enumerate}
\item The tetrahedra $|\simplex_{i-1}\cap \simplex_i|_i$ and $|\simplex_{i}\cap \simplex_{i+1}|_i$ form a thin wedge at the triangle $|\simplex_f|_i$. In this case (see figure \ref{fig:thin_wedge})
$$
N_{i+1\,i}=-\cosh(\dihedral_{f,i}) N_{i-1\,i}+\sinh(\dihedral_{f,i}) n_{i-1\,i},
$$
$$
n_{i+1\,i}=-\sinh(\dihedral_{f,i}) N_{i-1\,i}+\cosh(\dihedral_{f,i}) n_{i-1\,i},
$$
where $i\in\{1,\ldots,N\}$ and $N_{N+1\,N}:=N_{1\,N}, n_{N+1\,N}:=n_{1\,N}$.
Comparing with formula \eqref{eq:boostmatrixN} and \eqref{eq:boostmatrixn} we obtain:
$$
N_{i+1\,i}=-\exp(\frac{\dihedral_{f,i}}{\area_f}B_{f;i})N_{i-1\,i}.
$$
and 
$$
n_{i+1\,i}=\exp(\frac{\dihedral_{f,i}}{\area_f}B_{f;i})n_{i-1\,i}.
$$
\item The tetrahedra $|\simplex_{i-1}\cap \simplex_i|_i$ and $|\simplex_{i}\cap \simplex_{i+1}|_i$ form a thick wedge at the triangle $|\simplex_f|_i$. In this case (see figure \ref{fig:thick_wedge})
$$
N_{i+1\,i}=\cosh(\dihedral_{f,i}) N_{i-1\,i}-\sinh(\dihedral_{f,i}) n_{i-1\,i},
$$
$$
n_{i+1\,i}=\sinh(\dihedral_{f,i}) N_{i-1\,i}-\cosh(\dihedral_{f,i}) n_{i-1\,i}.
$$
Comparing with formula \eqref{eq:boostmatrixN} and \eqref{eq:boostmatrixn} we obtain:
$$
N_{i+1\,i}=\exp(\frac{\dihedral_{f,i}}{\area_f}B_{f;i})N_{i-1\,i},
$$
and 
$$
n_{i+1\,i}=-\exp(\frac{\dihedral_{f,i}}{\area_f}B_{f;i})n_{i-1\,i}.
$$
Let us note that $$-\exp(\frac{\dihedral_{f,i}}{\area_f}B_{f;i} + \frac{\pi}{\area_f}  \star B_{f;i})$$
is an SO(1,3) transformation fixing each vector parallel to the triangle $|\simplex_f|_i$ and mapping $n_{i-1\,i}$ into $n_{i+1\,i}$.
\end{enumerate}

Since we chose a gauge in which $g_{1\,2}=g_{2\,3}=\ldots=g_{N-1\,N}=\identity$, we have
$$
N_{i+1\,i}=-N_{i\, i+1}, n_{i+1\,i}=n_{i\,i+1}{\rm\ for\ }i\in\{1,\ldots,N-1\}.
$$
Let us define 
$$
\Pi_{f,i}:=\begin{cases} 0\textrm{, if the tetrahedra } |\simplex_{i-1}\cap \simplex_i|_i \textrm{ and }|\simplex_{i}\cap \simplex_{i+1}|_i\textrm{ form a thin wedge at the triangle }|\simplex_f|_i,\\
1\textrm{, if the tetrahedra } |\simplex_{i-1}\cap \simplex_i|_i \textrm{ and }|\simplex_{i}\cap \simplex_{i+1}|_i\textrm{ form a thick wedge at the triangle }|\simplex_f|_i.  \end{cases}
$$
Combining the observations above we conclude that
$$(-1)^{\sum_{i} \Pi_{f,i}}e^{\sum_i\left(\frac{\dihedral_{f,i}}{\area_f}B_{f;i} + \Pi_{f,i} \frac{\pi}{\area_f}  \star B_{f;i}\right)}$$
is the unique SO(1,3) transformation that fixes each vector parallel to the triangle $|\simplex_f|_i$ and maps $n_{N\,1}$ into $n_{1\,N}$. Therefore it coincides with $(g_{1\,N})^{-1}$. This means that
$$
\curvaturematrix_{f;1}=-\sum_i\left(\frac{\dihedral_{f,i}}{\area_f}B_{f;i} + \Pi_{f,i} \frac{\pi}{\area_f}  \star B_{f;i}\right)
$$
Since $\mu_i(v)=\mu_1(v),$ for $v\in\simplex_f$ and $\star B_{f;i}$ is orthogonal to $B_{f;i}$, we conclude
$$
\Tr{B_{f;1}\curvaturematrix_{f;1}}= 2\deficit_f \area_f.
$$
\end{proof}
\section{Summary, discussion and outlook}
For each pseudomanifold we constructed a smooth manifold, which we called a manifold with defects. In the standard approach a manifold with defects is obtained by triangulating a smooth manifold and removing simplices of dimension not exceeding $n-2$. Our construction is more general, because it includes all manifolds with defects obtained by the standard procedure but allows also manifolds with defects that cannot be completed to smooth manifolds without defects. Thanks to this our manifolds have purely combinatorial characterization. This has important technical and conceptual consequences. From technical point of view, it gives a better control over the histories of gravitational field that contribute to the path integral. From the conceptual point of view, it supports a scenario where a smooth space-time emerges from more fundamental combinatorial object as expected for example in \cite{Penrose}.

Our manifolds with defects are not simply connected. Therefore a holonomy of a flat connection around a closed loop can be non-trivial. As in the Aharonov-Bohm effect the curvature (describing magnetic field) has distributional support on the defect (describing thin and long solenoid). Since our manifolds in full generality cannot be completed to manifolds without defects these distributions have to be appropriately regularized. We regularized them by defining the action functional as a limit of Riemann-like sums. In our regularization we define (implicitly) the curvature at points of the gluing $\gluing$, where the metric is not smooth or even a neighbourhood of a point is not homeomorphic to an open subset of $\mathbb{R}^n$. It shares these features with the Lipschitz-Killing curvatures introduced in \cite{Cheeger, Wintgen, ChMSch} for Riemannian metrics. However, we focus on 4-dimensional Lorentzian Gravity in the Plebański formulation and we do not use the Lipschitz-Killing curvatures explicitly.

We imposed certain boundary conditions that correspond to Regge geometries. We showed that the field equations resulting from our action functional are
$$
D B =0,\quad \curvature=0
$$
on the manifold with defects. It turned out that our action functional evaluated at solutions of the field equations satisfying our boundary conditions coincides with the Regge action evaluated at the Regge variables defined by the boundary data. As a result the Hamilton-Jacobi functional for our action has an interpretation in terms of the Regge action. Therefore we expect that at the quantum level the theory defined by our action is a quantum version of Regge calculus. This quantum theory should be constructed using the spin-foam approach. The starting point should be a discrete action corresponding to a subdivision $\complex'_m$. We expect that a refinement limit $m\to\infty$ could be calculated exactly, because the spin-foam model of BF theory does not depend on a triangulation and we impose the constraints only at the boundary, not in the interior. The foams would be naturally embedded in the manifold with defects. Since each slice of such foam would be a spin-network state embedded in 3-dimensional manifold with defects and the curvature of the connection is concentrated on the defects, it seems possible that the spin foams constructed this way could define the dynamics of the Loop Quantum Gravity states in the recently introduced BF representation \cite{DittrichBFI,DittrichBFII,DittrichBFIII,DittrichBFIV,DittrichBFV}. 

Our approach provides new interpretation of the imposition of the constraints in the spin-foam models. General Relativity can be viewed as a BF theory with constraints \cite{Plebanski}. The constraints enforce equality of the right-handed area and the left-handed area defined by the B field \cite{ReisenbergerArea} for any surface embedded in space-time (see also \cite{GeneralizedSF}). We propose that if we impose the constraints only on the surface corresponding to the defects we obtain an approximation of General Relativity which coincides with the Regge theory. As a result, we propose an alternative interpretation of a single spin foam -- not as a truncation of the full theory obtained by discretization but rather by imposing the constraints only on certain surface, not on any surface. 

A relation between a theory of discrete gravity and topological field theory with curvature defects has been recently studied in \cite{WielandTQFT}. As in the model in \cite{WielandTQFT} we obtain a field theoretic description of a theory of discrete gravity. In our case this theory is precisely the Regge theory. In particular the curvature is concentrated on 2-dimensional cells whereas in \cite{WielandTQFT} it is concentrated on 3-dimensional cells. Another difference is that in \cite{WielandTQFT} the defects are light-like whereas in our work they are space-like. %We focus on flat geometries, which is the case of vanishing cosmological constant in \cite{WielandTQFT}.  

\section*{Acknowledgements}
I would like to thank Marios Christodoulou, Fabio D'Ambrosio, Klaus Liegener and Carlo Rovelli for interesting discussions. I am grateful for hospitality at the Centre de Physique Theorique de Luminy during my visit. This work has been partially supported by the grant of Polish Narodowe Centrum Nauki nr 2011/02/A/ST2/00300.
\bibliography{BFRegge}{}
\bibliographystyle{ieeetr}
\end{document}